\documentclass[letterpaper, journal, twoside]{IEEEtran}

\IEEEoverridecommandlockouts

\usepackage{graphics}
\usepackage{epsfig}
\usepackage{times}
\usepackage{amsmath}
\usepackage{stackengine}
\usepackage{amsthm}
\usepackage{multirow}
\usepackage{upgreek}
\usepackage{lettrine}
\usepackage{subfigure}
\usepackage{mathtools,amssymb,lipsum}
\usepackage{cuted}
 \usepackage{color}
\usepackage{cite}
\usepackage{cuted}
\usepackage{caption}
\usepackage{enumitem}
\usepackage{amsmath,amsfonts}
\usepackage{mathtools}
\usepackage{array}
\usepackage{textcomp}
\usepackage{stfloats}
\usepackage{url}
\usepackage{setspace}
\usepackage{graphicx}
\usepackage{nicematrix}
\usepackage{slashbox}

\newtheorem{theorem}{Theorem}

\newtheorem{lemma}{Lemma}
\newtheorem{assumption}{Assumption}
\newtheorem{definition}{Definition}
\newtheorem{remark}{Remark}

\newtheorem*{problem}{Problem}

\begin{document}
	
\title{On Event-Triggered Resilient Consensus Using Auxiliary Layer}

\author{Pushkal Purohit and Anoop Jain,~\IEEEmembership{Senior Member,~IEEE}\\
\thanks{The authors are with the Department of Electrical Engineering, Indian Institute of Technology Jodhpur, Rajasthan, India 342030 (e-mail: purohit.1@iitj.ac.in, anoopj@iitj.ac.in).}
}

\maketitle


\emph{This work has been submitted to the IEEE for possible publication. Copyright may be transferred without notice, after which this version may no longer be accessible.} \\

\begin{abstract}
Due to its design simplicity, \emph{auxiliary layer-based} resilient control is widely discussed in the literature to mitigate the effects of False Data Injection (FDI) attacks. However, the increased communication burden due to additional communication links for connecting an extra layer is often overlooked in the literature. This paper bridges this gap by considering an event-triggered approach for inter-layer communication between the physical layer (containing actual agents) and the auxiliary layer (containing virtual agents) for the resilient state consensus in a multi-agent system. We provide state-based and dynamic event-triggering mechanisms, the former being the motivation for the latter. The exclusion of Zeno behavior is established by proving positive minimum inter-event time (MIET). Extensive simulation and experimental results are provided to illustrate the proposed methodology. 
\end{abstract}



\section{Introduction} \label{sec:Intro}

\subsection{Motivation and Literature Review}
\lettrine{T}{he} advancement in communication technology and power-efficient computing has resulted in the wide use of distributed cyber-physical systems (CPSs) \cite{zhang2022advancements}. Though much designing and research go into primary protection to prevent any cyber attack on CPSs \cite{liu2022secure, zhu2020cross}, it is still susceptible to malicious attacks on communication channels \cite{dibaji2019systems}. Resilient control is used to keep the system running to achieve the desired goal, even in case of primary protection failure. Such an approach can be classified into two categories, viz. switching resilient control \cite{purohit2023passivity, yan2017resilient} and continuous operating resilient control \cite{meng2020adaptive, gusrialdi2018competitive}. This paper follows the latter approach of continuous operating resilient control irrespective of the presence of an attack, like other works in literature \cite{kuo2023resilient, usevitch2020resilient, huo2022observer, mustafa2020resilient}.

Some challenges in detecting and mitigating malicious attacks from the control theory perspective are discussed in \cite{tan2020brief,pirani2023graph}. These works mainly characterize the different types of attacks on a system and various methods to detect those attacks. The attack detection techniques can further be categorized into (i) data-based \cite{renganathan2021higher} and (ii) system-based approach \cite{purohit2023passivity}. The historical data is utilized in (i) to detect any discrepancies in the current system behavior, whereas approach (ii) exploits the knowledge of system dynamics to estimate the discrepancies in the system. However, cyber adversaries have evolved to design stealthy attacks \cite{wu2019optimal, zhao2023event} to remain undetected even in the presence of an attack detector. An adversary with complete knowledge of the system architecture can launch a stealthy attack designed to remain bounded and prevent the system from achieving the intended goal.

In response to such bounded stealthy attacks, it is sensible to have continuous operating resilient control. In this direction, an auxiliary layer-based approach was introduced in \cite{gharesifard2012resilience, gharesifard2013distributed, gusrialdi2014robust}, which provides resilience to cyber attacks on communication channels. This resilient control using an auxiliary network system falls under the category of \emph{network-based} resilient control \cite{pirani2023graph}. One of the crucial factors in designing such an auxiliary layer-based networked approach is the inter-layer interaction gain, which characterizes the competitive nature between the two layers (i.e., the physical layer and the auxiliary layer), and can be appropriately chosen by the designer to have a satisfactory resilience performance \cite{gharesifard2013distributed}. Due to its design simplicity, the concept of employing an extra layer has been further developed in literature \cite{gusrialdi2018competitive, iqbal2022distributed, kumar2022towards, gusrialdi2022cooperative, gusrialdi2023resilient, vaishnav2024auxiliary} while addressing the consensus problem in different application contexts. Specifically, \cite{gusrialdi2018competitive} discusses the resilient system design under undirected communication topology. The approach is extended to directed networks in \cite{iqbal2022distributed}. Both articles \cite{gusrialdi2018competitive} and \cite{iqbal2022distributed} consider homogeneous gain, which is extended to heterogeneous gain in \cite{kumar2022towards}, though only for two agents. Further, the auxiliary layer can be used for attack detection, as discussed in \cite{gusrialdi2022cooperative, vaishnav2024auxiliary}, where the redundant data received through the auxiliary and physical layer is compared to detect the presence of an attack. However, all these works \cite{gusrialdi2018competitive, iqbal2022distributed, kumar2022towards, gusrialdi2022cooperative, gusrialdi2023resilient, vaishnav2024auxiliary} overlook the fact that the presence of additional communication links for an extra layer increases the communication burden of the system. To reduce the communication burden, event-triggered communication has been widely studied in the literature \cite{nowzari2019event}.

\subsection{A Description of Our Approach}
Event-triggered communication can be categorized based on the trigger dependence, viz. (i) state-based, where events are triggered based on state; (ii) time-dependent, where events are triggered based on time; and (iii) dynamic, where events are triggered based on an internal dynamic variable. The problem of event-based secure consensus under FDI attacks has been previously studied in the literature \cite{wang2019resilient, zegers2021event}. Specifically, event-based consensus under the presence of malicious nodes is considered in \cite{wang2019resilient}, and resilience is achieved by filtering out the corrupted information received. A similar approach with the event and self-triggered control has been studied in \cite{zegers2021event} by detecting and isolating the adversarial agents. Unlike \cite{wang2019resilient, zegers2021event}, we do not isolate any agent of the system and consider an attack on the inter-agent communication instead of the presence of a malicious agent. The event-based communication method has also been applied to auxiliary layer-based resilient control in \cite{zhang2021event}; however, the event-based scheme is used only for communication from the physical layer to the auxiliary layer, whereas the continuous communication is used from the auxiliary layer to the physical layer. On the other hand, our work focuses on designing the event-triggered scheme for communication in both directions between the physical and auxiliary layers, which is more practical. As per the authors' knowledge, the dynamic triggering approach was first introduced in \cite{girard2014dynamic}, which is highly suitable for parallel operating event-based observers, as discussed in \cite{petri2021event}. Due to its affinity with the considered problem, we mainly focus on the dynamic triggering mechanism in this paper.

\subsection{Main Features and Contributions}
This paper considers the leader-follower state consensus problem in the presence of unknown FDI attacks affecting the agents and the underlying interaction network. The proposed solution methodology relies on constructing an auxiliary layer networked in conjunction with the physical layer, as shown in Fig.~\ref{fig:system}. The physical layer comprises a leader agent and remaining follower agents exposed to unknown FDI attacks. The auxiliary layer comprises the corresponding virtual agents not exposed to the attacker. We assume that the direct interaction of the leader with the agents in the physical and virtual layers cannot be attacked. However, the attacker may know the leader's state. Moreover, the physical/virtual agents in these layers are described using separate state models. With this background, we consider the event-triggered communication between the two layers where the state information transmitted to either of the layers is governed by an event detector, as shown in Fig.~\ref{fig:system}. The intra-layer communication among agents in both layers happens through the continuous communication channels, shown by solid lines in Fig.~\ref{fig:system}. Firstly, we design the state-based event-triggering condition, shown in Fig.~\ref{fig:layers} for the inter-layer communication between actual agents (with states $x_i, \forall i$) in the physical layer and virtual agents (with states $z_i, \forall i$) in the auxiliary layer where both layers receive and sends updated states at the same time $t_k$. We further extend the previous event design and propose a dynamic event-triggered sequence with separate event conditions for the two layers, using internal dynamical variables, as shown in Fig.~\ref{fig:layers2}. Under dynamic event-based communication, the physical layer sends the updated states ($x_i$) at time sequence $t_{k_x}$, whereas the auxiliary layer sends the updated states ($z_i$) at time sequence $t_{k_z}$. The practical state consensus is proved under both state-based and dynamic event conditions with a positive MIET. In summary, the paper contributes in the following aspects:
\begin{enumerate}[leftmargin=*]
	\item[(a)] Design of event-triggered inter-layer communication between physical and auxiliary layers, with the assurance of positive MIET for exclusion of Zeno behavior. Both state-based and dynamic event-triggering schemes are described under FDI attacks, the former being the motivation for the latter. 
	\item[(b)] Discussion on parameter selection and their effect on the event-triggered communication in auxiliary layer-based resilient control from a hardware implementation perspective. In particular, we highlight a trade-off between the inter-layer gain $\beta > 0$ and the frequency of events required for achieving the consensus in the vicinity of its steady state value under FDI attacks. It is observed that by increasing $\beta$, consensus is achieved in the closed proximity of its steady state value, with an expense of high frequency of events and smaller MIET.
	\item[(c)] Finally, we discuss how the auxiliary layer-based resilient control approach can be implemented on a physical system. In this direction, we conduct experiments on Khepera IV mobile robots by restricting their motion on the respective circular paths, thereby treating them as a single integrator model. With this setup, we aim to achieve the consensus in their angular position using the proposed approach. We also show how the large value of $\beta$ can result in actuator saturation and hence, need to be decided appropriately to fulfill both the physical and the event-triggering communication constraints.   
\end{enumerate}

\begin{figure}[t!]
	\centering{
	\subfigure[State-dependent events ]{\includegraphics[height=3.6cm]{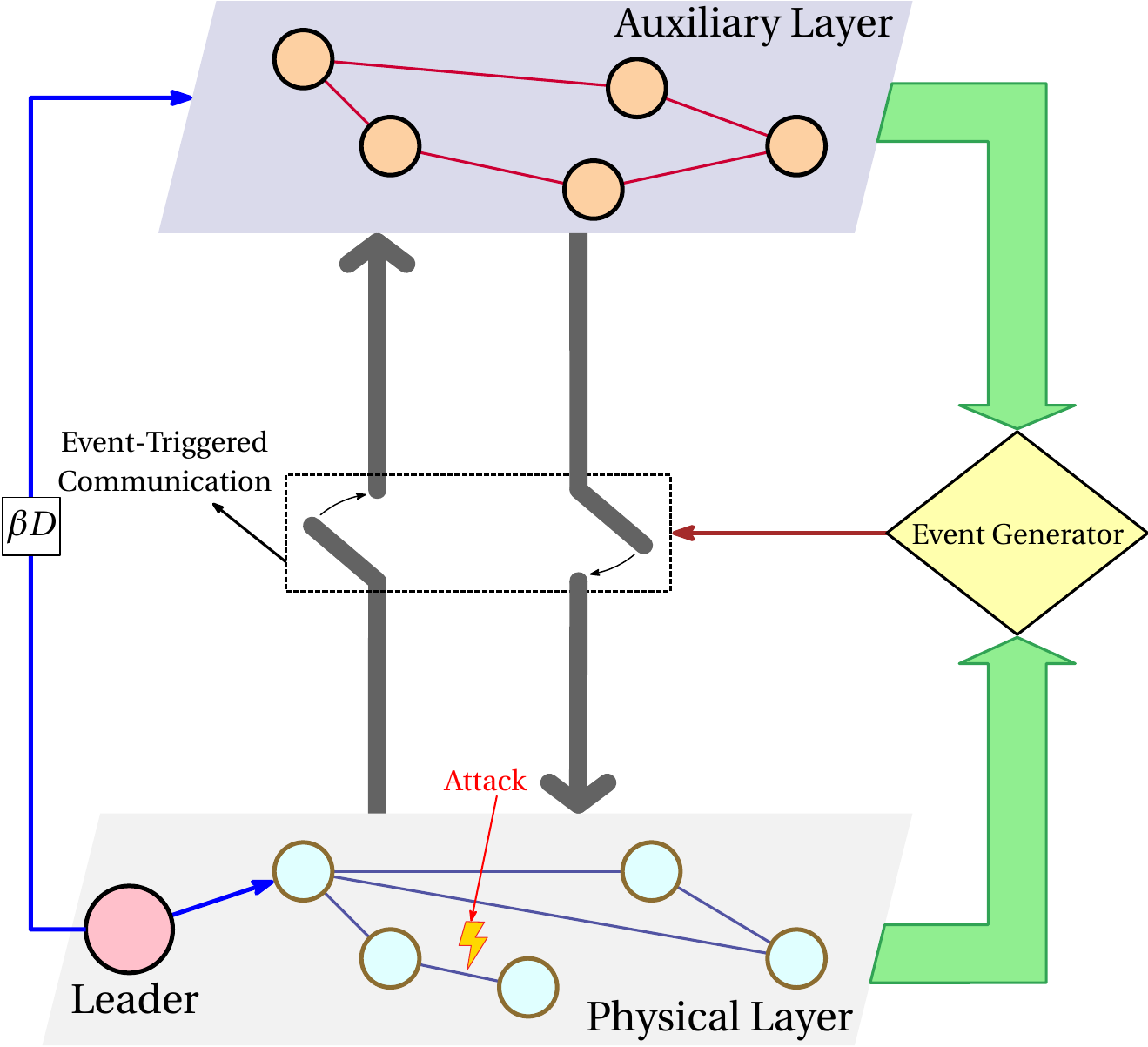}\label{fig:layers}} \hspace*{3pt}
	\subfigure[Layer-wise dynamic events ]{\includegraphics[height=3.6cm]{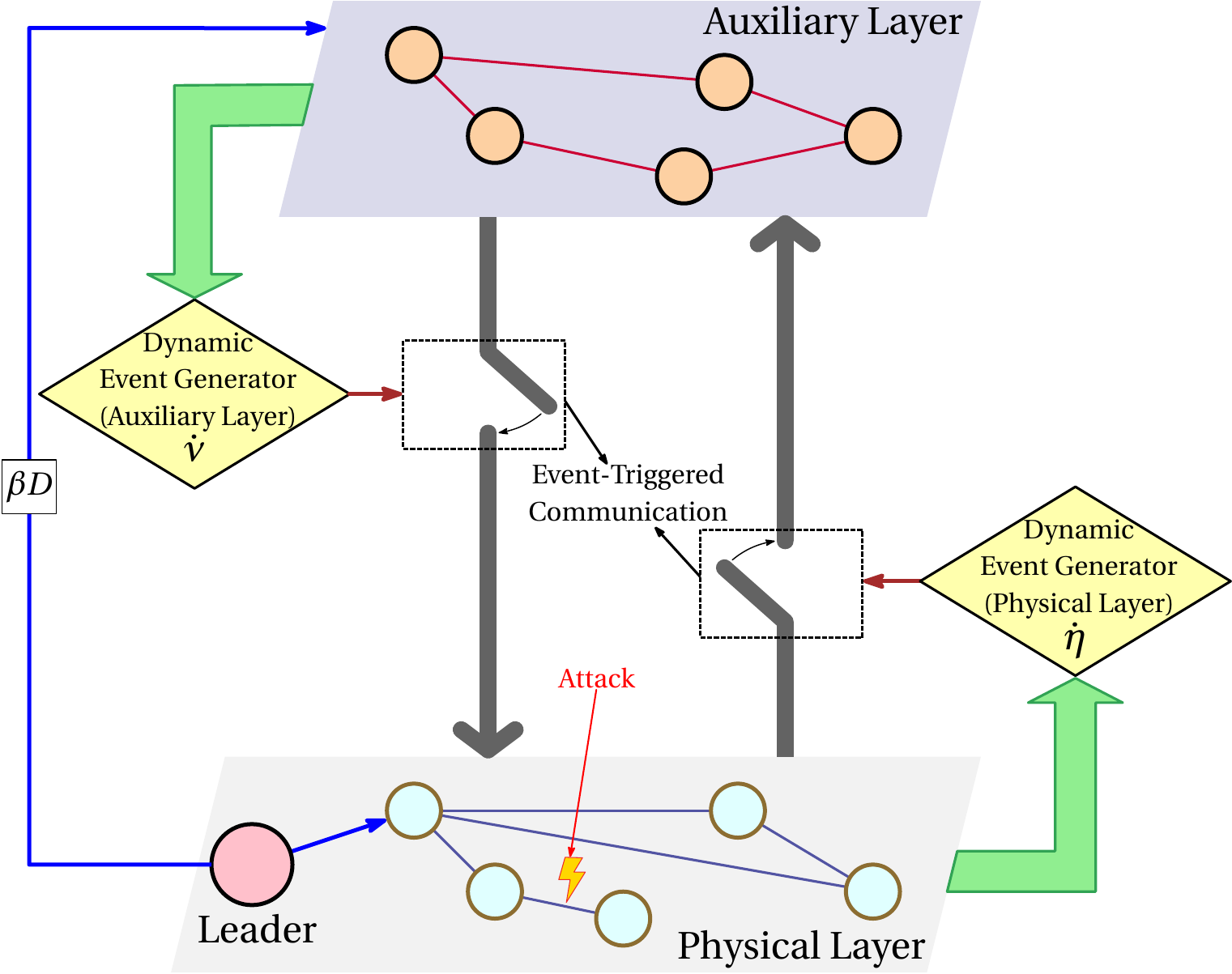}\label{fig:layers2}}
	\caption{Interaction between physical and auxiliary layers: (a) Single; (b) Separate; event generators for both layers.}
	\label{fig:system}}
	\vspace{-13pt}
\end{figure}

\subsection{Preliminaries}
\emph{Notations}: The set of real numbers and non-negative integers are represented by $\mathbb{R}$ and $\mathbb{Z}^+$, respectively. A vector of size $n$ with all entries $0$ (resp. $1$) is represented by $\pmb{0}_n$ (resp., $\pmb{1}_n$). An identity matrix of order $n$ is defined as $\pmb{I}_n$. A diagonal matrix is defined as $M \coloneqq \text{diag}\{m_i\}$ where $m_i \in \mathbb{R}$ denotes the entries along its principal diagonal. For a vector $\bullet \in \mathbb{R}^n$ (resp., a matrix $\star \in \mathbb{R}^{n \times n}$), the Euclidean norm (resp., Frobenius norm) is denoted by $\|\bullet\|$ (resp., $\|\star\|$). The transpose of a vector/matrix is represented using the superscript $\top$. The modulus of a scalar $q \in \mathbb{R}$ is represented by $|q|$. A positive (resp., negative) definite matrix is represented as $M \succcurlyeq 0$ (resp., $M \preccurlyeq 0$). The smallest eigenvalue of a square matrix $M \in \mathbb{R}^{n \times n}$ is given as $\lambda_{\text{min}}(M)$. The symbol $\land$ is used to represent the logical AND operation. The time argument $t$ is dropped at trivial places for clarity.

\emph{Graph Theory}: A graph is given as $\mathcal{G}=(\mathcal{V},\mathcal{E})$, where $\mathcal{V}$ is the set of nodes and $\mathcal{E}$ is the set of edges. The set of neighbors of node $i \in \mathcal{V}$ is denoted by $\mathcal{N}_i \coloneqq \{j \in \mathcal{V} \mid (j,i) \in \mathcal{E}\}$ with $|\mathcal{N}_i|$ being its cardinality. The Laplacian matrix is denoted as $\mathcal{L} = [\ell_{ij}] \in \mathbb{R}^{n \times n}$, where $\ell_{ij} = -1$ if $(j,i) \in \mathcal{E}$ and $\ell_{ij} = 0$ otherwise, for $j \neq i$, and $\ell_{ii} = - \sum_{j=1,j \neq i}^N \ell_{ij}$. In graph $\mathcal{G}$, a leader agent is defined as a set-point for the steady state consensus for the follower agents and is represented using a diagonal matrix $\mathcal{B} \coloneqq \text{diag}\{b_i\} \in \mathbb{R}^{n \times n}$, which corresponds to information flow from leader to followers, where $b_i = 1$ if node $i$ can directly receive information from leader and $b_i = 0$ otherwise. The follower agents are considered interacting according to an undirected graph, i.e., $(j, i) \in \mathcal{E} \iff (i,j) \in \mathcal{E}$. Throughout the paper, we assume that there exists a spanning tree with the leader being the root node, i.e., there is at least one $i$ for which $b_i =1$. Further, we consider that the follower agents within the physical and auxiliary layers interact according to an undirected and connected graph, as shown in Fig.~\ref{fig:system}.

\section{System and Attack Models and Problem Description} \label{sec:Prob and model}

\subsection{Physical Layer and Attack Model} \label{sec:system}
Consider the physical layer comprising a leader with static state $x_0 \in \mathbb{R}$, i.e., $\dot{x}_0 = 0$, and $N$ agents, interacting according to an undirected and connected topology with Laplacian $\mathcal{L}_x$, and governed by the following dynamics:
\begin{equation} \label{sys_singlelayer}
	\dot{x}_i(t) = \sum_{j \in \mathcal{N}_i} \ell_{ij} (x_j(t) - x_i(t)) + b_i (x_0 - x_i(t)),
\end{equation}
for all $i = 1, \ldots, N$, where $x_i \in \mathbb{R}$ is the state of the $i^{\text{th}}$ agent. The system \eqref{sys_singlelayer} can be written in the vector-matrix notation as:
\begin{equation} \label{sys_single_vector}
	\dot{x}(t) = A x(t) + B x_0,
\end{equation}
where $x = [x_1,\ldots,x_N]^\top \in \mathbb{R}^N$ is the state-vector for the group, the system matrix $A = -(\mathcal{L}_x + \mathcal{B}) \in \mathbb{R}^{N \times N}$ is a Hurwitz matrix \cite{vaishnav2024auxiliary} where the diagonal matrix $\mathcal{B} \coloneqq \text{diag}\{b_i\} \in \mathbb{R}^{N \times N}$ is defined in the preliminaries, and vector $B = \mathcal{B}\pmb{1}_N = [b_1,\ldots,b_N]^\top \in \mathbb{R}^N$. Next, we state the following definition of practical state consensus motivated from \cite[Difinition~2]{liu2020team}, before introducing FDI attacks: 

\begin{definition}[Practical Consensus] \label{def_practical_consensus}
	The system \eqref{sys_single_vector} is said to achieve the practical state consensus at the leader state $x_0$ asymptotically, if $\lim_{t \to \infty}|x_i(t) - x_0| \leq \Upsilon, \forall i$, where $\Upsilon > 0$ is an arbitrarily small constant.
\end{definition}

Under ideal conditions (i.e., in the absence of any attack), the system \eqref{sys_single_vector} achieves consensus at the leader state $x_0$ \cite[Theorem~1]{li2009leader}. However, an adversary can exploit the presence of a network to inject an external attack signal (usually known as an FDI attack) into the inter-agent communication channel. Let the injected attack signal from the $j^{\text{th}}$ agent to the $i^{\text{th}}$ agent in the physical layer be $d_{ij} \in \mathbb{R}$. Therefore, the corrupted signal received by the $i^{\text{th}}$ agent from the $j^{\text{th}}$ agent is given by:
\begin{equation} \label{attack}
	{x}_{ij} =\begin{cases}
		x_{j} + d_{ij}, & j \in \mathcal{N}_i\\
		0, & j \notin \mathcal{N}_i
	\end{cases}.
\end{equation}
Note that we consider a generic scenario where an attacker may inject different FDI signals between the two agents, that is, $d_{ji} \neq d_{ij}$, which is possible when the agents use separate channels for incoming and outgoing messages. This can be achieved using software-defined networking \cite{hu2014survey} and belongs to a class of Byzantine attacks \cite{pirani2023graph}. With the attack model \eqref{attack}, the system dynamics \eqref{sys_singlelayer} takes the form: $\dot{x}_i(t) = \sum_{j \in \mathcal{N}_i} \ell_{ij} ((x_j(t) + d_{ij}) - x_i(t)) + b_i (x_0 - x_i(t)) = \sum_{j \in \mathcal{N}_i} \ell_{ij} (x_j(t) - x_i(t)) + b_i (x_0 - x_i(t)) + d_i$, where $d_i \triangleq \sum_{j \in \mathcal{N}_i} \ell_{ij} d_{ij}$ is the joint effect of the multiple attacks (from all incoming signals) on the $i^{\text{th}}$ agent. The preceding expression can be written in the vector-matrix form as:
\begin{equation} \label{system_sld}
	\dot{x}(t) = A x(t) + B x_0 + d(x,t),
\end{equation}
where $d(x,t) \triangleq [d_1(x,t),\cdots,d_N(x,t)]^\top \in \mathbb{R}^N$ is the attack vector. In the subsequent analysis, we consider the following assumption on the attack vector:
\begin{assumption} \label{assumption}
	The FDI attack vector $d(x,t)$ is uniformly bounded for any bounded $x$, that is, $\|d(x,t)\| \leq \bar{D}, \forall t \geq 0$ where $\bar{D} > 0$ is a known constant. In particular, if the attack is governed by the dynamics $\dot{d} = f(d, x)$, this must be finite-gain $\mathcal{L}_2$ stable.
\end{assumption}

\begin{remark} \label{rem_attack}
Note that Assumption~\ref{assumption} is practical in the sense that any unbounded attack signal can be easily identified, in contrast to a bounded attack signal that may remain stealthy and destabilize the system \cite{gusrialdi2018competitive, zhang2021stealthy}. Further, it is important to highlight the fact that the noise and disturbances in a system can also be modeled similarly to the attack model \eqref{attack}; for instance, please refer to \cite[Section~II]{cheng2016on}. However, the main difference is that the noise and disturbances are random environmental signals, while the attack signals are deliberately injected to prevent the system from achieving its goal while remaining undetected; please refer to \cite{pirani2023graph, taheri2020undetectable}. In this paper, we consider no noise/disturbance is present in the system.
\end{remark}

\subsection{Auxiliary Layer and Resultant System} \label{sec:aux_layer_sys}
To obtain resilience against FDI attacks, a parallel operating auxiliary layer of virtual agents can be interconnected to the physical layer, as shown in Fig.~\ref{fig:system}. As per the construction, we consider a corresponding virtual agent in the auxiliary layer for each follower agent in the physical layer. Further, the interaction topology among agents in both layers is not necessarily the same and can be described using the different undirected and connected communication graphs. With this background, the dynamics of the resultant system, under the presence of an FDI attack in \eqref{attack}, is given by \cite{gharesifard2012resilience, gusrialdi2018competitive}:
\begin{subequations} \label{system_init}
	\begin{align}
		\dot{x}(t) & = Ax(t) + \beta Kz(t) + Bx_0 + d(x,t), \label{system_init_x}\\
		\dot{z}(t) & = Hz(t) - \beta Gx(t) + \beta Dx_0, \label{system_init_z}
	\end{align}
\end{subequations}
where $z_i \in \mathbb{R}, \forall i$, $z=[z_1,\ldots,z_N]^\top \in \mathbb{R}^N$ is the state vector of the virtual agents. Further, the matrices $H, K, G \in \mathbb{R}^{N \times N}$ and vector $D \in \mathbb{R}^N$ characterize the interconnection between the two layers and need to be designed appropriately. Specifically, the matrix $H$ governs the intra-layer communication of virtual agents within the auxiliary layer, $K$ and $G$ represent communication from the auxiliary layer to the physical layer and vice versa, and the vectors $B, D \in \mathbb{R}^{N}$ represent the communication from the leader to the followers in physical and auxiliary layers, respectively. It is worth noticing that the case $\beta = 0$ is the same as the absence of the auxiliary layer since \eqref{system_init_x} reduces to \eqref{system_sld}, with \eqref{system_init} being a decoupled system. Below, we describe the selection of these matrices to support the subsequent analysis.

\subsubsection*{Selection of Matrices}\label{subsubsec:matrix_selection} 
As mentioned above, the matrix $A$ is Hurwitz, and $H$ can be chosen as a Hurwitz matrix; hence, there exist symmetric positive definite matrices $P_x$ and $P_z$ such that $A^\top P_x + P_x A \preccurlyeq 0$ and $H^\top P_z + P_z H \preccurlyeq 0$. When the underlying graph network contains a leader agent as the root node, one of the methods for calculating matrix $P_x$ is given as \cite[Theorem~1]{zhang2015constructing}:
\begin{subequations}\label{lyap_matrix}
	\begin{align}
		\zeta & = [\zeta_1,\cdots,\zeta_N]^\top = A^{-1} \pmb{1}_N,\\
		\chi & = [\chi_1,\cdots,\chi_N]^\top = A^{-\top} \pmb{1}_N,\\
		P_x & = \text{diag}\{\chi_i/\zeta_i\}, \ \forall i.
	\end{align}
\end{subequations}
Similarly, $P_z$ can be calculated as above by replacing $A$ with $H$. Further, one can choose $H = -(\mathcal{L}_z + \mathcal{W})$, where $\mathcal{L}_z$ is the Laplacian matrix of the virtual agents in the auxiliary layer, and $\mathcal{W} \triangleq \text{diag}\{w_i\} \in \mathbb{R}^{N \times N}$ with at least one $w_i \neq 0$. If the matrix $H$ is selected as a Laplacian matrix $\mathcal{L}_z$ as in \cite{gusrialdi2014robust}, the auxiliary layer states $z_i$ in \eqref{system_init_z} converge at some arbitrary constant value rather than origin, because of the zero eigenvalue of the Laplacian $\mathcal{L}_z$. It is worth noticing that, although the states $z_i$ of the virtual agents in the auxiliary layer have no practical significance, and hence, the equilibrium where these converge. However, such a choice of $H$ is not useful in designing the event sequence in our formulation and requires $H$ to be chosen as a Hurwitz matrix. We also choose the interconnection matrices $K$ and $G$ as invertible Hurwitz matrices. Such a choice of matrices $A, H, K$, and $G$ ensures that the auxiliary states $z_i$ converge in the neighborhood of the origin and also helps design the event condition. 

Although these matrices can be designed arbitrarily, as per our above discussion, one can further restrict these choices by analyzing the special case of no attack condition for the system \eqref{system_init}. In this direction, the following lemma establishes a relation between the matrices $K$ and $G$ and column vector $D$ such that the agents in the physical layer converge to the leader state $x_0$: 

\begin{lemma}[\hspace{-0.1pt}{\cite[Lemma~1]{gusrialdi2018competitive}}] \label{matrix_cond}
	In the absence of external attack signal $d(x,t)$ (i.e., $d(x,t) \equiv \pmb{0}_N$), the physical system states $x_i$ achieve consensus at $x_0$, under dynamics \eqref{system_init}, i.e., $x_i(t) \to x_0, \forall i$ as $t \to \infty$, if $K$, $G$ and $D$ are chosen as:
	\begin{equation} \label{condition_matrices}
		K^\top P_x = P_z G, \qquad	D = G \pmb{1}_N.
	\end{equation}
\end{lemma}

From the hardware implementation viewpoint, the following remark is useful regarding the selection of the above matrices:
\begin{remark} \label{rem_matrix_selection}
	Let $i_x$ and $i_z$ be the $i^{\text{th}}$ node in physical and auxiliary layers, respectively. Let $\mathcal{N}_{i_x}$ be the set of neighbors of $i_x$ in the physical layer, and $\mathcal{N}_{i_z}$ be the neighbors of $i_z$ in the auxiliary layer. The selection of $K=H$ and $G=A$ is crucial from a practical implementation point of view, which indicates that the nodes $i_x$ and $i_z$ can receive information of agents in $\mathcal{N}_{i_x} \cup \mathcal{N}_{i_z}$. One trivial case to select such matrices while satisfying the condition \eqref{condition_matrices} is $H=K=G=A$, with $P_x$ and $P_z$ being identity matrices.
\end{remark}
 
\subsection{Problem Description} \label{sec:problem}
Consider the corresponding agents in the physical and auxiliary layers share the state information only at event times $t_k$, where $k \in \mathbb{Z}^+$. Let the associated sampled states be denoted by $\bar{x}(t) = x(t_k)$ and $\bar{z}(t) = z(t_k), \ t \in [t_k, t_{k+1}), \forall k$ for the agents in the physical and auxiliary layers, respectively. Under this event-triggered sampling, the system \eqref{system_init} can be written as:
\begin{subequations}\label{system_event_triggered}
	\begin{align}
		\dot{x}(t) & = Ax(t) + \beta K\bar{z}(t) +  Bx_0 + d(x, t) \\
		\dot{z}(t) & = Hz(t) - \beta G\bar{x}(t) + \beta Dx_0,
	\end{align}
\end{subequations}
where $t \in [t_k, t_{k+1}), \forall k$. Let $e_x(t) \triangleq x(t) - \bar{x}$ and $e_z(t) \triangleq z(t) - \bar{z}$ be the sampling errors for $t \in [t_k, t_{k+1}),\forall k$. The system \eqref{system_event_triggered}, in terms of sampling errors $e_x(t)$ and $e_z(t)$, can be expressed as:
\begin{subequations} \label{system}
	\begin{align}
		\dot{x}(t) & = Ax(t) + \beta Kz(t) +  Bx_0 - \beta Ke_z(t) + d(x, t) \label{system_x} \\
		\dot{z}(t) & = Hz(t) - \beta Gx(t) + \beta Dx_0 + \beta G e_x(t), \label{system_z}
	\end{align}
\end{subequations}
where $t \in [t_k, t_{k+1}), \forall k$. Further, \eqref{system} can be written in terms of consensus error $\tilde{x}(t) \triangleq x(t) - x_0 \pmb{1}_N$ (with respect to the leader) as:
\begin{align*}
	\dot{\tilde{x}}(t) & = A \tilde{x}(t) {+} \beta K z(t) {+} B x_0 {-} \beta Ke_z(t) {+} d(x, t) {+} A \pmb{1}_N x_0 \\
	\dot{z}(t) & = H z(t) - \beta G \tilde{x}(t) + \beta Dx_0 + \beta G e_x(t) - \beta G \pmb{1}_N x_0,
\end{align*}
where $t \in [t_k, t_{k+1}), \forall k$. By substituting $A \pmb{1}_N x_0 = -(\mathcal{L}_x + \mathcal{B})\pmb{1}_N x_0 = - B x_0$ from Subsection~\ref{sec:system} and $G \pmb{1}_N = D$ from Lemma~\ref{matrix_cond}, the above equations are simplified as: 
\begin{subequations} \label{err_system}
	\begin{align}
		\dot{\tilde{x}}(t) & = A\tilde{x}(t) + \beta Kz(t) - \beta Ke_z(t) + d(x,t)\\
		\dot{z}(t) & = Hz(t) - \beta G\tilde{x}(t) + \beta G e_x(t),
	\end{align}
\end{subequations}
where $t \in [t_k, t_{k+1}), \forall k$. In the absence of a FDI attack (i.e., $d \equiv \pmb{0}_N$), the equilibrium, denoted by $[\tilde{x}^e, z^e]^\top$, of the system \eqref{err_system} satisfies (since $e_x = e_z = 0$ at the equilibrium): $\pmb{0}_N = A\tilde{x}^e + \beta Kz^e; \pmb{0}_N = Hz^e - \beta G\tilde{x}^e$ with the origin $[\pmb{0}_N,\pmb{0}_N]^\top$ being the only solution for Hurwitz matrices $A,H,K$ and $G$ (see \cite{gusrialdi2018competitive}). We now formally describe the problem addressed in this paper: 
\begin{problem}\label{problem}
Consider the system shown in Fig.~\ref{fig:system}, where the corresponding agents in the physical and auxiliary layers interact according to the event-triggering mechanism and have continuous intra-layer communication among them in the respective layer. Let the agents in the physical layer be subject to (unknown) FDI attacks $d(x, t)$ satisfying Assumption~\ref{assumption}. Design an event-triggering condition for updating sampled values $\bar{x}$ and $\bar{z}$ for the communication between the physical and auxiliary layers in Fig.~\ref{fig:layers} and Fig.~\ref{fig:layers2} such that the equilibrium of the system \eqref{err_system} is stable in the sense of Definition~\ref{def_practical_consensus}. That is, the states $x_i$ of the actual agents in the physical layer achieve practical consensus at the leader state $x_0$, and the states $z_i$ of the virtual agents in the auxiliary layer achieve practical consensus at the origin, i.e., $\lim_{t \to \infty}|x_i(t)-x_0| \leq \vartheta$ and $\lim_{t \to \infty} |z_i(t)| \leq \varphi$, $\forall i$, where $\vartheta, \ \varphi >0$ are arbitrary small constants.	
\end{problem}

Moreover, it will be shown that the consensus bound ($\vartheta$) depends on the parameters of the event condition and the inter-layer gain $\beta$. Specifically, we provide a discussion analyzing the effect of parameter selection on the number of events, MIET, and consensus bound. 

\section{Design of Inter-layer Event Mechanism and Stability Analysis} \label{sec:stability}
This section focuses on designing the event condition for the inter-layer communication between the agents in the physical and auxiliary layers (Fig.~\ref{fig:system}) such that the equilibrium of \eqref{err_system} is asymptotically stable in the sense of Definition~\ref{def_practical_consensus}. Further, the existence of positive MIET is proven under the obtained event conditions, followed by a discussion about the selection of design parameters and their influence on the proposed event condition and MIET. We first discuss the case of state-based event-triggering in Fig.~\ref{fig:layers}, followed by the dynamic event-triggering-based approach as in Fig.~\ref{fig:layers2}.

\subsection{State-Based Event-Triggering} \label{sec:event}
For the system \eqref{err_system} to have the stable equilibrium $[\tilde{x}^e, z^e]^\top$ at the origin, we consider the quadratic candidate Lyapunov function $V(\tilde{x}(t), z(t)) = \tilde{x}^\top(t) P_x \tilde{x}(t) + z^\top(t) P_z z(t)$, where matrices $P_x$ and $P_z$ are defined in Subsection~\ref{subsubsec:matrix_selection}. The time derivative of $V(\tilde{x}(t), z(t))$ is $\dot{V} = 2 \tilde{x}^\top P_x \dot{\tilde{x}} + 2 z^\top P_z \dot{z}$. Substituting for $\dot{\tilde{x}}$ and $\dot{z}$ from \eqref{err_system}, we get
\begin{align*}
	\dot{V} & = 2 \tilde{x}^\top P_x (A\tilde{x} + \beta Kz - \beta Ke_z + d) \\
	& \qquad + 2 z^\top P_z (Hz - \beta G\tilde{x} + \beta G e_x) \\
	& = 2 \tilde{x}^\top P_x A\tilde{x} + 2\beta \tilde{x}^\top P_x K z - 2\beta \tilde{x}^\top P_x Ke_z + 2\tilde{x}^\top P_x d \\
	& \qquad + 2 z^\top P_z Hz - 2 \beta z^\top P_z G \tilde{x} + 2 \beta z^\top P_z G e_x.
\end{align*}
Notice that $2x^\top P_x A x = x^\top (P_x A + A^\top P_x) x$ and $2 z^\top P_z Hz = z^\top (P_z H + H^\top P_z) z$, implying that
\begin{align*}
&	\dot{V}  = \tilde{x}^\top (P_x A {+} A^\top P_x) \tilde{x} + z^\top (P_z H {+} H^\top P_z) z + 2\tilde{x}^\top P_x d \\
	& + 2\beta \tilde{x}^\top P_x K z {-} 2 \beta z^\top P_z G \tilde{x} {-} 2\beta \tilde{x}^\top P_x Ke_z {+} 2 \beta z^\top P_z G e_x.
\end{align*}
Now, using the fact that the scalar $2 \beta z^\top P_z G \tilde{x} = (2 \beta z^\top P_z G \tilde{x})^\top = 2 \beta \tilde{x}^\top G^\top P_z z$, we can further simplify as:
\begin{align} \label{lyap_deri0}
	\dot{V} & = -\tilde{x}^\top Q_x \tilde{x} {-} z^\top Q_z z {+} 2\tilde{x}^\top P_x d {+} 2\beta \tilde{x}^\top (P_x K {-} G^\top P_z) z \nonumber \\
	& \qquad - 2\beta \tilde{x}^\top P_x Ke_z + 2 \beta z^\top P_z G e_x,
\end{align}
where $Q_x \triangleq -(P_x A + A^\top P_x) \succcurlyeq 0$ and $Q_z \triangleq -(P_z H + H^\top P_z) \succcurlyeq 0$ are positive definite matrices as described in Subsection~\ref{subsubsec:matrix_selection}. Substituting $P_x K = G^\top P_z$ from Lemma~\ref{matrix_cond} into \eqref{lyap_deri0} and using Rayleigh inequality \cite[Theorem~4.2.2]{horn2012matrix}, we get
\begin{align*}
	\dot{V} & \leq -\lambda_{\text{min}} (Q_x) \|\tilde{x}\|^2 - \lambda_{\text{min}} (Q_z) \|z\|^2 + 2\|\tilde{x}^\top P_x d\| \\
	& \qquad + 2\beta \|\tilde{x}^\top P_x K e_z\| + 2 \beta \|z^\top P_z G e_x\|.
\end{align*}
Now, using sub-multiplicative inequality $\|\mathcal{X} \mathcal{Y}\| \leq \|\mathcal{X}\| \|\mathcal{Y}\|$ for the square matrices $\mathcal{X}, \mathcal{Y} \in \mathbb{R}^{n \times n}$ \cite[Section~5.2]{meyer2000matrix}, we have
\begin{align} \label{lyap_deri}
	\dot{V} & \leq -\lambda_{\min} (Q_x) \|\tilde{x}\|^2 - \lambda_{\min} (Q_z) \|z\|^2 + 2\|\tilde{x}\| \|P_x\| \|d\| \nonumber \\
	& \qquad + 2\beta \|\tilde{x}\| \|P_x K\| \|e_z\| + 2 \beta \|z\| \|P_z G\| \|e_x\|.
\end{align}
Further, using the inequality $2XY \leq cX^2 + ({1}/{c})Y^2$, for any $X,Y,c>0$, the last three terms in \eqref{lyap_deri} can be upper bounded as: $2\|\tilde{x}\| \|P_x\| \|d\| \leq c_3 \|\tilde{x}\|^2 + ({\|P_x\|^2}/{c_3}) \|d\|^2$, $2 \beta \|z\| \|P_z G\| \|e_x\| \leq c_2 \beta \|z\|^2 + ({\beta}/{c_2}) \|P_z G\|^2 \|e_x\|^2$, and $2\beta \|\tilde{x}\| \|P_x K\| \|e_z\| \leq c_1 \beta \|\tilde{x}\|^2 + ({\beta}/{c_1}) \|P_x K\|^2 \|e_z\|^2$ for some $c_1, c_2, c_3 > 0$. Using these inequalities, we infer from \eqref{lyap_deri} that
\begin{align} \label{lyap_deri1}
	\dot{V} & \leq -\lambda_{\text{min}} (Q_x) \|\tilde{x}\|^2 - \lambda_{\text{min}} (Q_z) \|z\|^2 + c_1 \beta \|\tilde{x}\|^2 + c_3 \|\tilde{x}\|^2 \nonumber \\
	& \qquad + c_2 \beta \|z\|^2 + \frac{\beta}{c_1} \|P_x K\|^2 \|e_z\|^2 + \frac{\beta}{c_2} \|P_z G\|^2 \|e_x\|^2 \nonumber \\
	& \qquad + \frac{\|P_x\|^2}{c_3} \|d\|^2.
\end{align}
Following Assumption~\ref{assumption}, we use the fact that $\|d(x, t)\|^2 \leq \bar{D}^2$, to yield:
\begin{align} \label{lyap_deriF}
	\dot{V} & \leq -\left(\lambda_{\text{min}} (Q_x) {-} c_3 {-} c_1 \beta \right) \|\tilde{x}\|^2 {-} \left(\lambda_{\text{min}} (Q_z) {-} c_2 \beta\right) \|z\|^2 \nonumber\\
	& \quad {+} \frac{\beta}{c_1} \|P_x K\|^2 \|e_z\|^2 {+} \frac{\beta}{c_2} \|P_z G\|^2 \|e_x\|^2 {+} \frac{\|P_x\|^2}{c_3} \bar{D}^2,
\end{align}
where $(c_3 + c_1 \beta) \in (0,\lambda_{\min} (Q_x))$ and $c_2 \beta \in (0,\lambda_{\min}(Q_z))$ are design parameters. From \eqref{lyap_deriF}, it readily follows that the equilibrium of \eqref{err_system} is Lyapunov stable if the following condition holds:
\begin{align} \label{event_cond_single_layer}
	& \frac{\beta}{c_1} \|P_x K\|^2 \|e_z\|^2 + \frac{\beta}{c_2} \|P_z G\|^2 \|e_x\|^2 \leq \nonumber \\
	& \qquad \kappa \biggl( (\lambda_{\text{min}} (Q_x) - c_3 - c_1 \beta) \|\tilde{x}\|^2 \nonumber \\
	& \qquad + (\lambda_{\text{min}} (Q_z) - c_2 \beta) \|z\|^2 + \epsilon - \frac{\|P_x\|^2}{c_3} \bar{D}^2 \biggr),
\end{align}
where $\kappa \in (0,1)$ is a design variable, and $\epsilon > ({\|P_x\|^2}/{c_3}) \bar{D}^2$ is a constant with $P_x$ being the solution to Lyapunov equation (see \eqref{lyap_matrix}), and is used to avoid the Zeno behavior. Based on the condition \eqref{event_cond_single_layer}, the event-triggering sequence can be defined as:
\begin{equation} \label{event2}
\hspace*{-11pt}	\begin{aligned}
		t_0 & = 0, \\
		t_{k+1} & {=} \inf \left\{\begin{array}{l}
			t \in \mathbb{R} \mid t>t_k \ \land \\ \frac{\beta}{c_1} \|P_x K\|^2 \|e_z\|^2 {+} \frac{\beta}{c_2} \|P_z G\|^2 \|e_x\|^2 \geq \\
			\kappa \biggl( (\lambda_{\min} (Q_x) {-} c_3 {-} c_1 \beta) \|\tilde{x}\|^2\\
			{+} (\lambda_{\min} (Q_z) {-} c_2 \beta) \|z\|^2 {+} \epsilon {-} \frac{\|P_x\|^2}{c_3} \bar{D}^2 \biggr)		
		\end{array}\right\}.
	\end{aligned}
\end{equation}
Further, the presence of a positive MIET under \eqref{event2} is proven in the Appendix. It is worth noting that the implementation of \eqref{event2} requires continuous monitoring of the signals $e_x(t)$ and $\tilde{x}(t)$ from the physical layer and $e_z(t)$ and $z(t)$ from the auxiliary layer to generate an event. This requires all these signals to be continuously available at one processing point for comparison and event generation, as shown in Fig.~\ref{fig:layers} via a single event generator. However, implementing such an event condition \eqref{event2} defeats the purpose of the event-triggered approach \cite{petri2021event}. Indeed, it is more practical that the physical layer signals ($e_x(t)$ and $\tilde{x}(t)$) and auxiliary layer signals ($e_z(t)$ and $z(t)$) be monitored separately to generate an event for communication. In this direction, we utilize the internal dynamics-based event-triggering approach \cite{girard2014dynamic} to design two separate event conditions for physical and auxiliary layers, as shown in Fig.~\ref{fig:layers2}, which is the topic of the following subsection.

\subsection{Internal Dynamics-Based Event-Triggering} \label{sec:event_2}
Let the events for the corresponding agents' interaction in the physical and auxiliary layers be generated at time instants $t_{k_x}$ and $t_{k_z}$, respectively, where $k_x, k_z \in \mathbb{Z}^+$. With these updated event times, let the sampled values be defined as $\bar{x}(t) = x(t_{k_x})$ for $t \in [t_{k_x}, t_{k_x + 1}), \forall t_{k_x}$ and $\bar{z}(t) = z(t_{k_z})$ for $t \in [t_{k_z}, t_{k_z + 1}), \forall t_{k_z}$. And the corresponding sampling errors be defined as $e_x(t) = x(t) - \bar{x}$ for $t \in [t_{k_x}, t_{k_x + 1}), \forall t_{k_x}$ and $e_z(t) = z(t) - \bar{z}$ for $t \in [t_{k_z}, t_{k_z + 1}), \forall t_{k_z}$. Next, we propose the following internal dynamical variables $\eta$ and $\nu$, dependent on the sampling errors $e_x$ and $e_z$, satisfying the following differential equations:
\begin{subequations} \label{aux_dynamics}
	\begin{align}
		\dot{\eta}(t) & = - \eta(t) + \sigma_1 \|e_x(t)\|^2 \label{aux_dynamics_x}, \qquad \eta(0) = \eta_0,\\
		\dot{\nu}(t) & = - \nu(t) + \sigma_2 \|e_z(t)\|^2 \label{aux_dynamics_z}, \qquad \nu(0) = \nu_0,
	\end{align}
\end{subequations}
for all $t \geq 0$, where $\sigma_1, \sigma_2$ are positive design constants. From \cite[Lemma~2.2]{girard2014dynamic}, it follows that the solution of \eqref{aux_dynamics} satisfies $\eta(t) \geq 0$ and $\nu(t) \geq 0, \ \forall t \geq 0$ for the initial values $\eta_0 \geq 0$ and $\nu_0 \geq 0$. Based on the dynamics \eqref{aux_dynamics}, consider that the agents in physical layer sends their states $x_i, \forall i$ at event sequence defined as:
\begin{subequations} \label{event_separate}
	\begin{equation} \label{event_cond_physical}
		\begin{aligned}
			t_{0_x} & = 0, \\
			t_{k_x+1} & = \inf \left\{\begin{array}{l}
				t \in \mathbb{R} \mid t>t_{k_x} \ \land \\ \|e_x\|^2 \geq \dfrac{c_2 (\eta + \Omega)}{\beta \|P_z G\|^2 + \sigma_1 c_2}
			\end{array}\right\},
		\end{aligned}
	\end{equation}
	and the virtual agents in auxiliary layer sends their states $z_i, \forall i$ at event sequence defined as:
	\begin{equation} \label{event_cond_auxiliary}
		\begin{aligned}
			t_{0_z} & = 0, \\
			t_{k_z+1} & = \inf \left\{\begin{array}{l}
				t \in \mathbb{R} \mid t>t_{k_z} \ \land \\ \|e_z\|^2 \geq \dfrac{c_1 (\nu + \mu)}{\beta \|P_x K\|^2 + \sigma_2 c_1}
			\end{array}\right\},
		\end{aligned}
	\end{equation}
\end{subequations}
where $\Omega = \varepsilon - ({\|P_x\|^2}/{c_3})\bar{D}^2>0$, $\varepsilon > ({\|P_z G\|^2}/{c_3})\bar{D}^2>0$, and $\mu > 0$ are constants, and $c_1, c_2, c_3$ are as defined in \eqref{lyap_deriF}. The below remark briefly discusses the selection of the design variables $\Omega$ and $\mu$, before we prove the practical state consensus of the system \eqref{err_system} under the event condition \eqref{event_separate}.

\begin{remark}
Note that the constant $\Omega$ in physical layer event condition \eqref{event_cond_physical} is used to eliminate the positive attack term $({\|P_z G\|^2}/{c_3})\bar{D}^2$ in \eqref{colyap_deri}, while guaranteeing the presence of positive MIET, as discussed later in the paper. However, the same can be assured by putting the term with auxiliary layer event condition by interchanging the constant terms $\Omega$ and $\mu$ in \eqref{event_separate}, i.e., selecting $\Omega>0$ and $\mu = \varepsilon - ({\|P_z G\|^2}/{c_3})\bar{D}^2$. The choice depends on whether the physical or auxiliary layer should handle (relatively) more events, as compared to the other. Note that $P_z$ is given analogously to $P_x$ in \eqref{lyap_matrix}. 
\end{remark}

We now state the following theorem proving the practical consensus of the system \eqref{err_system} under the proposed event-triggering condition \eqref{event_separate}:

\begin{theorem} \label{thm_stability_event}
Under the separate layer-based event conditions \eqref{event_separate} where the dynamics of internal variables $\eta$ and $\nu$ is governed by \eqref{aux_dynamics} with $\eta_0, \nu_0 \geq 0$, the equilibrium $[\tilde{x}^e, z^e]^\top$ of the system \eqref{err_system}, as shown in Fig.~\ref{fig:layers2}, is asymptotically stable in the sense of Definition~\ref{def_practical_consensus}, that is, $|x_i(t) - x_0| \leq \vartheta$ and $|z_i(t)| \leq \varphi$, $\forall i, j$ as $t \to \infty$ for some arbitrary small constants $\vartheta, \varphi > 0$.
\end{theorem}

\begin{proof}
	Consider the joint candidate Lyapunov function $U(\tilde{x}(t),z(t),\eta(t),\nu(t)) = V(\tilde{x}(t),z(t)) + \eta(t) + \nu(t) \geq 0$, where $V$ is defined in Section~\ref{sec:event}, and $\eta$ and $\nu$ are governed by the dynamics \eqref{aux_dynamics}. The time derivative of $U$ is given by $\dot{U} = \dot{V} + \dot{\eta} + \dot{\nu}$ where substituting for $\dot{V}$, $\dot{\eta}$ and $\dot{\nu}$ from \eqref{lyap_deriF} and \eqref{aux_dynamics}, respectively, and combining terms, we get
	\begin{align} \label{colyap_deri}
		\dot{U}	& {\leq} {-} \left(\lambda_{\text{min}} (Q_x) {-} c_3 {-} c_1 \beta \right) \|\tilde{x}\|^2 {+} \left( \frac{\beta\|P_z G\|^2 {+} \sigma_1 c_2}{c_2} \right) \|e_x\|^2 \nonumber \\
		& \quad {-} \left(\lambda_{\text{min}} (Q_z) {-} c_2 \beta\right) \|z\|^2 {+} \left( \frac{\beta\|P_x K\|^2 {+} \sigma_2 c_1}{c_1} \right) \|e_z\|^2 \nonumber \\
		& \quad - \eta - \nu + \frac{\|P_x\|^2}{c_3} \bar{D}^2.
	\end{align}
	From \eqref{event_separate}, we have $\|e_x(t)\|^2 \leq (c_2 (\eta + \Omega))/(\beta \|P_z G\|^2 + \sigma_1 c_2)$ and $\|e_z(t)\|^2 \leq (c_1 (\nu + \mu))/(\beta \|P_x K\|^2 + \sigma_2 c_1)$ for all $t \geq 0$. Substituting these into \eqref{colyap_deri}, we obtain $\dot{U}	\leq -\left(\lambda_{\text{min}} (Q_x) - c_3 - c_1 \beta \right) \|\tilde{x}\|^2 - (\lambda_{\min} (Q_z) {-} c_2 \beta) \|z\|^2 + \Omega + \mu + ({\|P_x\|^2}/{c_3}) \bar{D}$. Further, replacing $\Omega = \varepsilon - ({\|P_x\|^2}/{c_3})\bar{D}^2$ as discussed in \eqref{event_separate}, we have
	\begin{align}\label{stability}
		\dot{U} & \leq -\left(\lambda_{\text{min}} (Q_x) {-} c_3 {-} c_1 \beta \right) \|\tilde{x}\|^2 {-} \left(\lambda_{\text{min}} (Q_z) {-} c_2 \beta\right) \|z\|^2 \nonumber \\
		& \quad + \varepsilon - \frac{\|P_x\|^2}{c_3} \bar{D} + \mu + \frac{\|P_x\|^2}{c_3} \bar{D}, \nonumber \\
		& = -\left(\lambda_{\text{min}} (Q_x) {-} c_3 {-} c_1 \beta \right) \|\tilde{x}\|^2 {-} \left(\lambda_{\text{min}} (Q_z) {-} c_2 \beta\right) \|z\|^2 \nonumber \\
		& \quad + \varepsilon + \mu. 
	\end{align}
	From \eqref{stability}, it can be observed that $\dot{U} < 0$, for $(\lambda_{\text{min}} (Q_x) - c_3 - c_1 \beta) \|\tilde{x}\|^2 + (\lambda_{\text{min}} (Q_z) - c_2 \beta) \|z\|^2 > \varepsilon + \mu$. Alternatively, this means that as $t \to \infty$ the states $\tilde{x}$ and $z$ converge to an elliptical region defined by $\|\tilde{x}\|^2/g_1^2 + \|z\|^2/g_2^2 < 1$, where $g_1 = \sqrt{(\varepsilon + \mu)/(\lambda_{\min} (Q_x) {-} c_3 {-} c_1 \beta)}$  and $g_2 = \sqrt{(\varepsilon + \mu)/(\lambda_{\min} (Q_z) {-} c_2 \beta)}$. Note that $g_1$ and $g_2$ are well-defined since $\varepsilon > 0$, $\mu > 0$ from \eqref{event_separate}, and $\lambda_{\min} (Q_x) {-} c_3 {-} c_1 \beta > 0$, $\lambda_{\min} (Q_z) {-} c_2 \beta > 0$ from \eqref{lyap_deriF}. As a result, it can be concluded that $\|\tilde{x}(t)\| < g_1$ and $\|z(t)\| < g_2$ as $t \to \infty$. Further, since $\dot{U} < 0$, we have $U$ decreasing over time and it follows that $\lim_{t \to \infty} U \in \{\hbar \in \mathbb{R} \mid 0 \leq \hbar \leq \varPhi\}$ for some small positive constant $\varPhi$ \cite[Theorem~4.18]{khalil2002nonlinear}. From the definition of $U$ and selection of $\eta_0 \geq 0$ and $\nu_0 \geq 0$ in \eqref{aux_dynamics}, we have that $V(x(t)) \leq U(x(t),\eta(t), \nu(t))$ for all $t \geq 0$. In other words, $V$ is upper bounded by $U$ for all $t \geq 0$, and hence, $V$ also converges to the neighborhood of origin, i.e., $\lim_{t \to \infty} V \in \{\hbar \in \mathbb{R} \mid 0 \leq \hbar \leq \varPsi\}$, where $\varPsi \leq \varPhi$ is a small positive constant. Consequently, it follows that there exist small constants $\vartheta, \varphi > 0$ such that $|x_i(t) - x_0| \leq \vartheta$ and $|z_i(t)| \leq \varphi$, $\forall i, j$ as $t \to \infty$. This completes the proof.
\end{proof}

To show that there exists a positive MIET for both physical and auxiliary layers' events, it is considered that the state derivatives (i.e., $\dot{x}$ and $\dot{z}$) are bounded, i.e., $\| \dot{x} \| = \|A \tilde{x} + \beta K z + \beta Ke_z + d\| \leq M_x$ and $\| \dot{z} \| =  \|H z - \beta G \tilde{x} + \beta G e_x\| \leq M_z$, where $M_x$ and $M_z$ are arbitrarily large positive constants. Consequently, the sampling errors $e_x = x - \bar{x}$, $e_z = z - \bar{z}$ have bounded derivatives i.e., $\|\dot{e}_x\| = \| \dot{x} \| \leq M_x$ and $\| \dot{e}_z \| = \| \dot{z} \| \leq M_z$. Note that this assumption is intrinsic in all stable systems and arises due to hardware limitations (for example a speed limit of a robot) and has been common in the literature \cite{petri2021event,zegers2021event}, since unbounded dynamics ($\dot{x}$ and $\dot{z}$) would mean the system is unstable, which also means that a strictly positive MIET does not exist. The following theorem shows the presence of a positive MIET under the above assumption.

\begin{theorem}[Existence of MIET] \label{thm_MIET}
Under the conditions given in Theorem~\ref{thm_stability_event}, the event condition \eqref{event_separate} does not exhibit Zeno behavior and there exists a positive MIET for both physical and auxiliary layers, as given below:
	\begin{enumerate}[leftmargin=*]
		\item[(a)] MIET for physical layer event condition \eqref{event_cond_physical} is
		\begin{equation} \label{tau_x}
			\tau_{x} \geq \frac{1}{2M_x} \sqrt{\frac{c_2 \Omega}{\beta \|P_z G\|^2 + \sigma_1 c_2}} > 0.
		\end{equation}
		\item[(b)] MIET for auxiliary layer event condition \eqref{event_cond_auxiliary} is
		\begin{equation} \label{tau_z}
			\tau_z \geq \frac{1}{2M_z} \sqrt{\frac{c_1 \mu}{\beta \|P_x K\|^2 + \sigma_2 c_1}} > 0.
		\end{equation}
	\end{enumerate}
\end{theorem}

\begin{proof}
For brevity, we provide the proof only for part (a), and the proof of part (b) follows accordingly. For the MIET ($\tau_x$) in the physical layer, the event triggers at time $t_{k_x} + \tau_x$ when the right side of \eqref{event_cond_physical} is minimum, i.e.,
	\begin{equation*}
		\|e_x(t_{k_x} + \tau_x)\|^2 = \min \left\{\frac{c_2 (\eta + \Omega)}{\beta \|P_z G\|^2 + \sigma_1 c_2}\right\},
	\end{equation*}
	where $\eta$ is the only variable. Since $\eta_0 \geq 0$, we have from \eqref{aux_dynamics_x} that $\min_{t} \eta(t) = 0$, and hence,
	\begin{equation} \label{error_MIET_physical}
		\|e_x(t_{k_x} + \tau_x)\|^2 = \frac{c_2 \Omega}{\beta \|P_z G\|^2 + \sigma_1 c_2}.
	\end{equation}
	Further, it can be obtained that
	\begin{equation} \label{error_dynamics_physical}
		\frac{d}{dt} \|e_x(t)\|^2 {=} \frac{d}{dt} e_x^\top e_x = 2 e_x^\top \dot{e_x} \leq 2 \|e_x(t)\| M_x.
	\end{equation}
	For the event-triggering at MIET, substitute for $\|e_x\|$ from \eqref{error_MIET_physical} into \eqref{error_dynamics_physical}, we get
	\begin{equation*}
		\frac{d}{dt} \|e_x(t)\|^2 \leq 2M_x \sqrt{\frac{c_2 \Omega}{\beta \|P_z G\|^2 + \sigma_1 c_2}}.
	\end{equation*}
	Integrating the above for $t \in [t_{k_x}, t_{k_x} + \tau_x)$, we get $\|e(t_{k_x})\|^2 + \|e_x(t_{k_x} + \tau_x)\|^2 \leq 2M_x \tau_x \sqrt{\frac{c_2 \Omega}{\beta \|P_z G\|^2 + \sigma_1 c_2}}$. From definition of sampling error $\|e(t_{k_x})\|^2 = 0$, we have $\|e_x(t_{k_x} + \tau_x)\|^2 \leq 2M_x \tau_x \sqrt{\frac{c_2 \Omega}{\beta \|P_z G\|^2 + \sigma_1 c_2}}$. Rearranging and substituting from \eqref{error_MIET_physical}, we get \eqref{tau_x}. Following the steps similar to part~(a), we get MIET for auxiliary layer $\tau_z$ as in \eqref{tau_z}. This concludes the proof.
\end{proof}

From \eqref{tau_x} and \eqref{tau_z}, it can be observed that the large values of $\beta$ result in smaller MIET ($\tau_x$ and $\tau_z$), thus requiring a relatively large number of events to achieve consensus. A further discussion on the selection of $\beta$ is provided in the upcoming section.

\begin{figure}[t!]
	\centering{
		\includegraphics[width=0.7\linewidth]{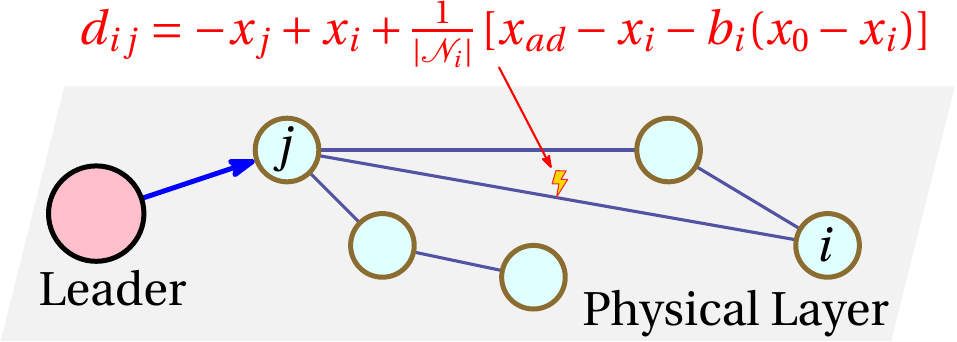}
		\caption{FDI attack on the link connecting $i^\text{th}$ and $j^\text{th}$ agents.}
		\label{fig:attack_example_illust}}
	\vspace{-13pt}
\end{figure}

\begin{figure}[t!]
	\centering
	\includegraphics[width=0.7\linewidth]{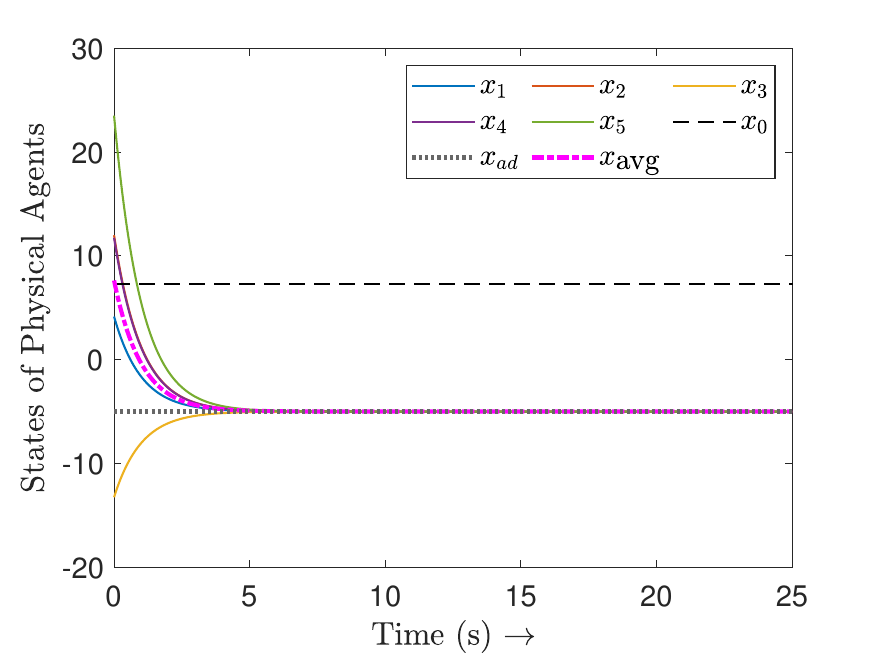}
	\caption{Physical states $x_i$ under attack in absence of auxiliary layer.}
	\label{fig:single_layer}
	\vspace{-15pt}
\end{figure}

\begin{figure*}[t!]
	\centering
	\subfigure[Physical states $x_i$]{\includegraphics[width=0.38\textwidth]{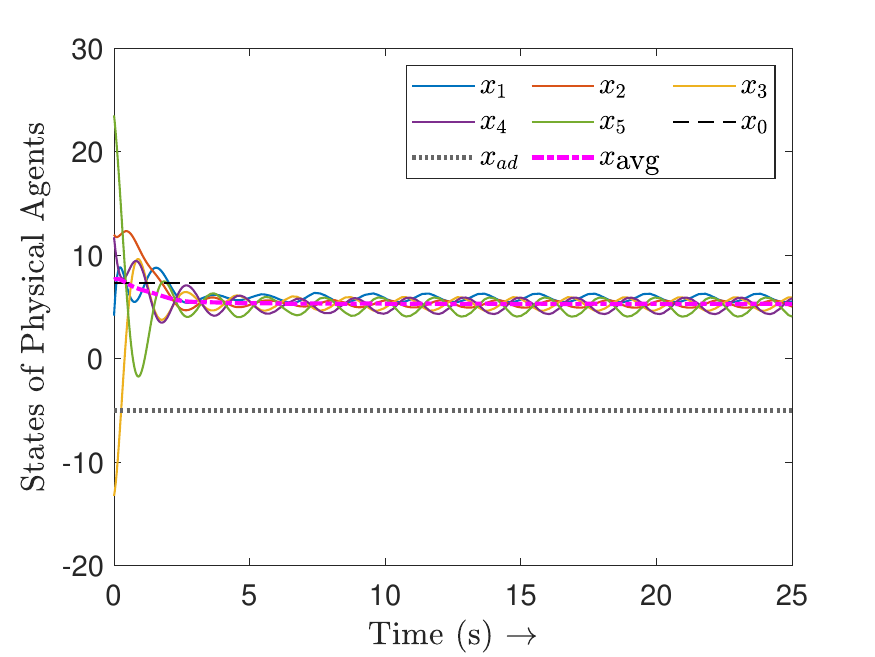}\label{fig:Sim_x1}} \hspace*{1pt}
	\subfigure[Auxiliary states $z_i$]{\includegraphics[width=0.38\textwidth]{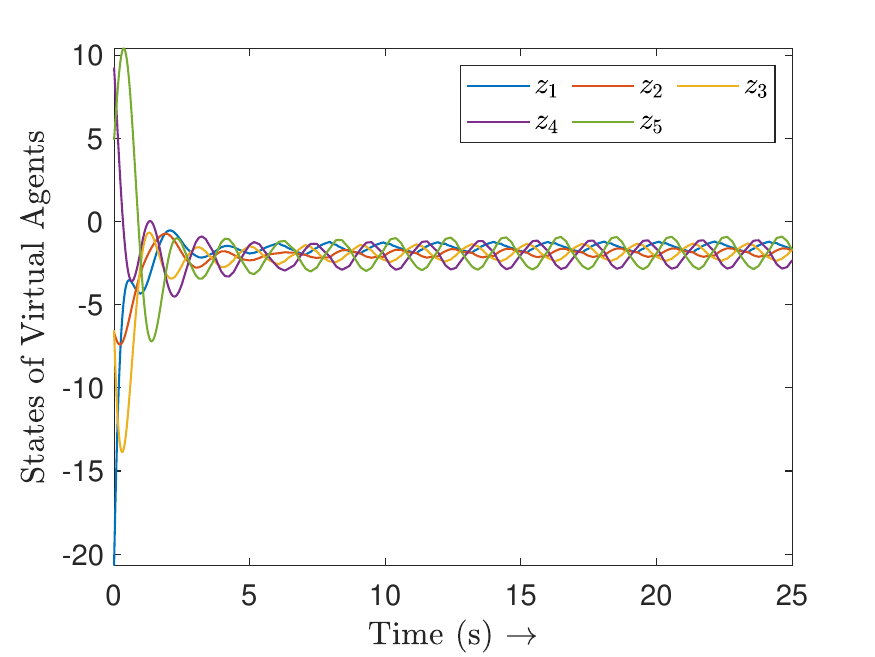}\label{fig:Sim_z1}}
	\subfigure[Events generated in both layers]{\includegraphics[width=0.8\textwidth]{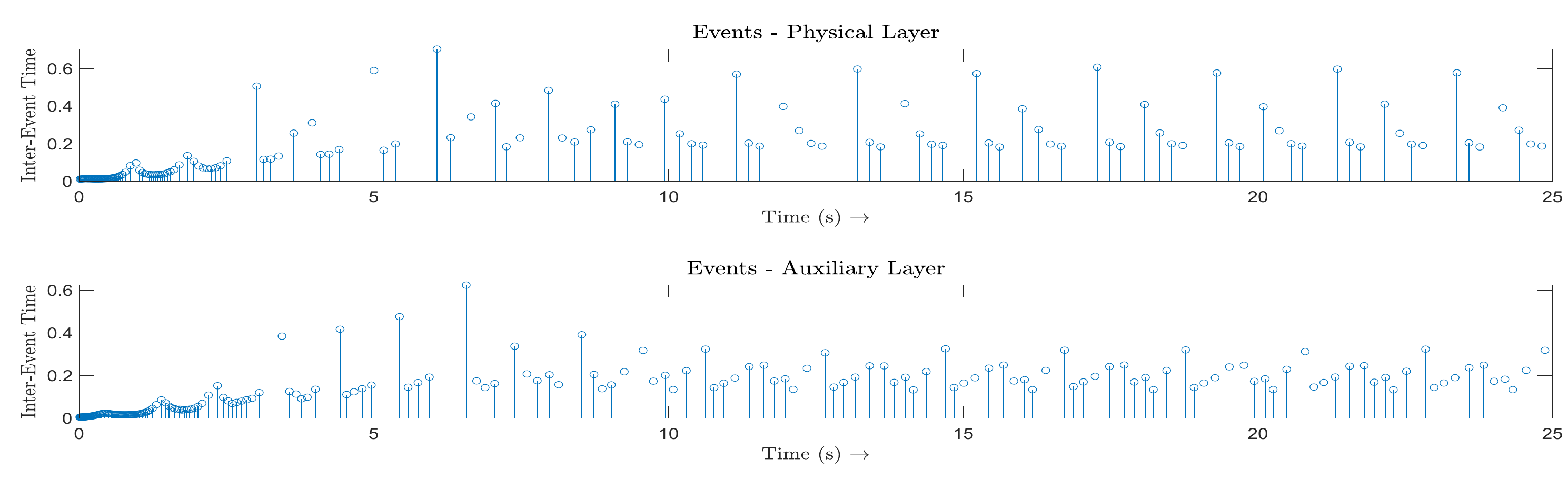}\label{fig:Sim_event1}}
	\caption{Physical and virtual agents' states, and the event sequence for system \eqref{err_system} with $\beta = 1$.}
	\label{fig:Sim_1}
	\vspace{-13pt}
\end{figure*}

\section{Simulations and Experiments}

\subsection{Simulation Example} \label{sec:Sim}
Consider a network of five agents in the physical layer subject to FDI attacks, as shown in Fig.~\ref{fig:attack_example_illust} where the attack is depicted for the link connecting the $i^\text{th}$ and $j^\text{th}$ agents. The leader's state is chosen to be an arbitrary constant value, $x_0 = 7.33861$. The cumulative FDI attack on the $i^{\text{th}}$ agent is given by $d_i=\sum_{j \in \mathcal{N}_i} d_{ij} = -\sum_{j \in \mathcal{N}_i} (x_j - x_i) + \sum_{j \in \mathcal{N}_i} ({x_{ad} - x_i - b_i (x_0 - x_i)})/{|\mathcal{N}_i|} = -\sum_{j \in \mathcal{N}_i} (x_j - x_i) + x_{ad} - x_i - b_i (x_0 - x_i), \forall i$, where $x_{ad} = -5$ is a (finite) offset decided by the attacker. In matrix form the attack vector is given as $d = -Ax + \pmb{1}_N x_{ad} - x - B x_0$, which satisfies Assumption~\ref{assumption}. Under this attack, \eqref{sys_singlelayer} becomes $\dot{x}_i = \sum_{j \in \mathcal{N}_i} (x_j - x_i) + b_{i} (x_0 - x_i) -\sum_{j \in \mathcal{N}_i} (x_j - x_i) + x_{ad} - x_i - b_i (x_0 - x_i) = x_{ad} - x_i$. The overall system can be expressed in the form \eqref{system_sld}, where the associated Laplacian $\mathcal{L}_x$ and $B$ are given by:
$$\mathcal{L}_x = \begin{bmatrix}
	3 & -1 & -1 &  0 & -1\\
	-1 &  2 & -1 &  0 &  0\\
	-1 & -1 &  3 & -1 &  0\\
	0 &  0 & -1 &  1 &  0\\
	-1 &  0 &  0 &  0 &  1
\end{bmatrix}, \quad B = \begin{bmatrix}
	1\\
	0\\
	0\\
	0\\
	0
\end{bmatrix}.$$
From the above we get:
$$A = -(\mathcal{L}_x + \text{diag}\{b_i\}) = \begin{bmatrix}
	-4 &  1 &  1 &  0 &  1\\
	1 & -2 &  1 &  0 &  0\\
	1 &  1 & -3 &  1 &  0\\
	0 &  0 &  1 & -1 &  0\\
	1 &  0 &  0 &  0 & -1
\end{bmatrix}.$$
\begin{itemize}[leftmargin=*]
\item A simulation example in the absence of the auxiliary layer is shown in Fig.~\ref{fig:single_layer}. It can be seen that all agents converge at the value $x_{ad} = -5$. Note that the controller does not have any knowledge about the attack signal $d$, it is assumed in this form only for illustration.

\item In view of Remark~\ref{rem_matrix_selection}, we further select $H=A$ and $K=A$. From \eqref{lyap_matrix}, we get $P_x = \pmb{I}_N$ and $P_z = \pmb{I}_N$, and from Lemma~\ref{matrix_cond} we have $G = P_z^{-1} K P_x = K$ and $D = G\pmb{1}_N = [-1, 0, 0, 0, 0]^\top$. Further, we calculate $Q_x = -(P_x A + A^\top P_x)$ and $Q_z = -(P_z H + H^\top P_z)$ as described in \eqref{lyap_deri0} to get $\lambda_{\min}(Q_x) = 0.48967$ and $\lambda_{\min} (Q_z) = 0.30900$. Other parameters used for simulation are $\sigma_1 = \sigma_2 = 1$, $\varepsilon = 186.39439$, and $\mu = 0.27775$. The constants $c_1$, $c_2$ and $c_3$ are selected randomly to satisfy the condition $(c_3 + c_1 \beta) \in (0,\lambda_{\text{min}} (Q_x))$ and $c_2 \beta \in (0,\lambda_{\text{min}}(Q_z))$, as defined for \eqref{lyap_deriF}. Simulations are carried out using Euler's method with the step size of $dt = 0.00001 s$.

\item Simulation results for the resilient system \eqref{err_system} comprising an auxiliary layer, with events defined by \eqref{event_separate} are depicted in Fig.~\ref{fig:Sim_1} and Fig.~\ref{fig:Sim_25} for different values of inter-layer gain $\beta$. In these figures, part (a) shows the states $x_i$ in the physical layer, part (b) shows the states $z_i$ in the auxiliary layer, and part (c) shows the instances of events with inter-event times for physical and auxiliary layers \eqref{event_separate}. It is observed that the higher value of $\beta$ results in increased event frequency and a large amount of oscillation in the response. To reduce these oscillations, we need to keep the value of $\beta$ relatively small. For comparison, smallest inter-event time, from simulation, for both layers is summarized in Table~\ref{MIET_tble} for different values of $\beta$. It can be observed that the smallest inter-event time decreases for both layers with an increase in $\beta$.

\begin{table}[t!]
	\centering
	\caption{Smallest inter-event time versus $\beta$}
	\label{MIET_tble}
	\begin{NiceTabular}{|l||c|c|c|}
		\hline
		\backslashbox{\textbf{Layer $\downarrow$}}{\textbf{$\beta \rightarrow$}} & \textbf{1} & \textbf{2.5} & \textbf{5}\\
		\hline
		\hline
		\textbf{Physical Layer} & 0.01243 & 0.00449 & 0.00201\\
		\hline
		\textbf{Auxiliary Layer} & 0.00566 & 0.00301 & 0.00155\\
		\hline
	\end{NiceTabular}
	\vspace{-15pt}
\end{table}

\item Since the states of the physical agents oscillate in the steady state (see part (a) of Fig.~\ref{fig:Sim_1} and Fig.~\ref{fig:Sim_25}), it is difficult to observe the effect of the varying inter-layer gain $\beta$. In this direction, we defined the mean state as $x_{\text{avg}}(t) = ({1}/{N}) \sum_i^N x_i(t), \forall t \geq 0$ shown by dash-dotted line in Fig.~\ref{fig:single_layer} and part (a) of Fig.~\ref{fig:Sim_1} and Fig.~\ref{fig:Sim_25}. The difference between the leader and the mean value of the agents' states (i.e., $x_0 - x_{\text{avg}}(t)$) for four different values of $\beta$ is shown in Fig.~\ref{fig:Error}, where $\beta=0$ represents the absence of the auxiliary layer. It can be seen in Fig.~\ref{fig:Error} that by increasing the value of inter-layer gain $\beta$, the value $x_0 - x_{\text{avg}}(t)$ becomes relatively smaller as $t \to \infty$.
\end{itemize}

\begin{figure*}[t!]
	\centering
	\subfigure[Physical states $x_i$]{\includegraphics[width=0.38\textwidth]{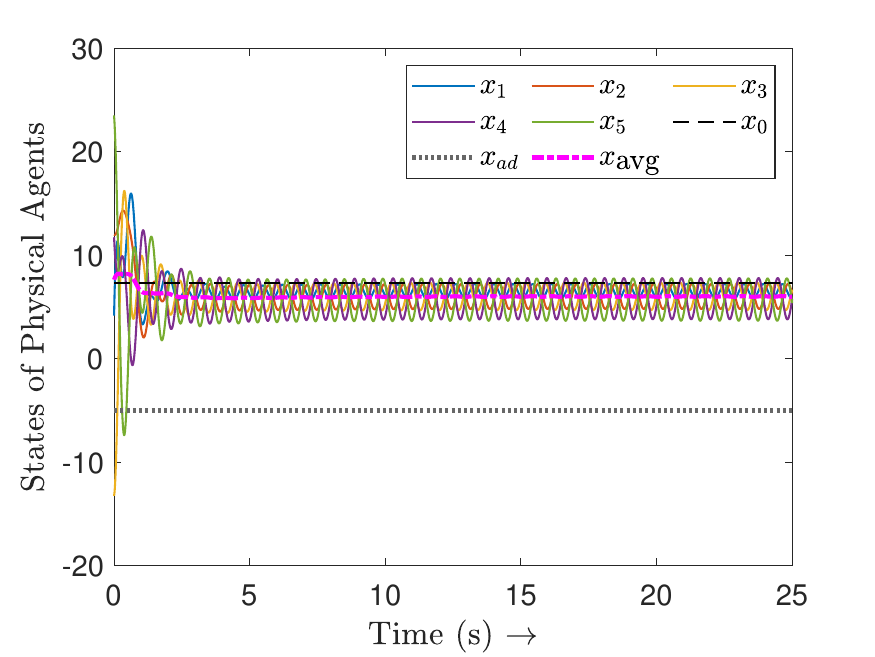}\label{fig:Sim_x25}} \hspace*{1pt}
	\subfigure[Auxiliary states $z_i$]{\includegraphics[width=0.38\textwidth]{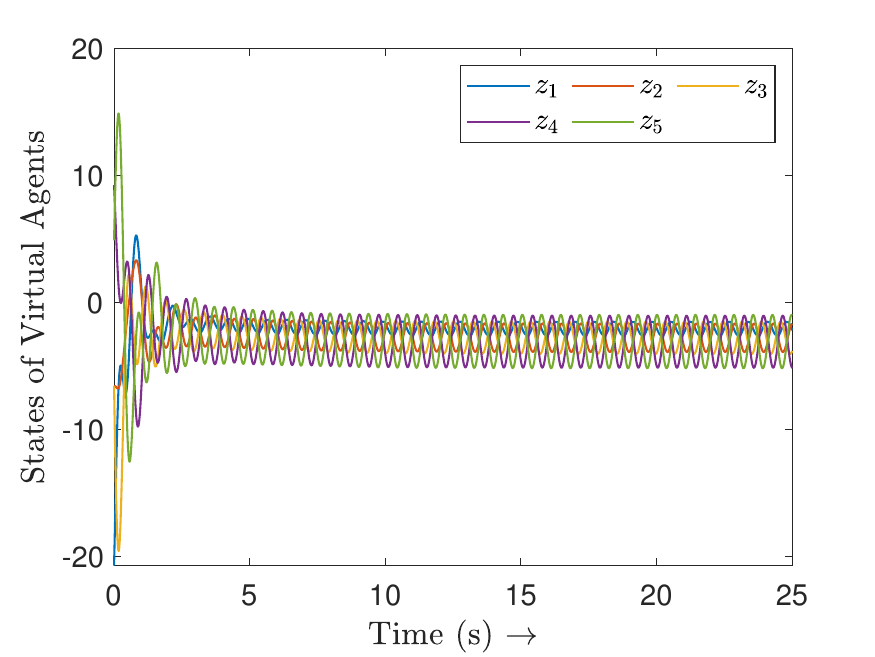}\label{fig:Sim_z25}}
	\subfigure[Events generated in both layers]{\includegraphics[width=0.8\textwidth]{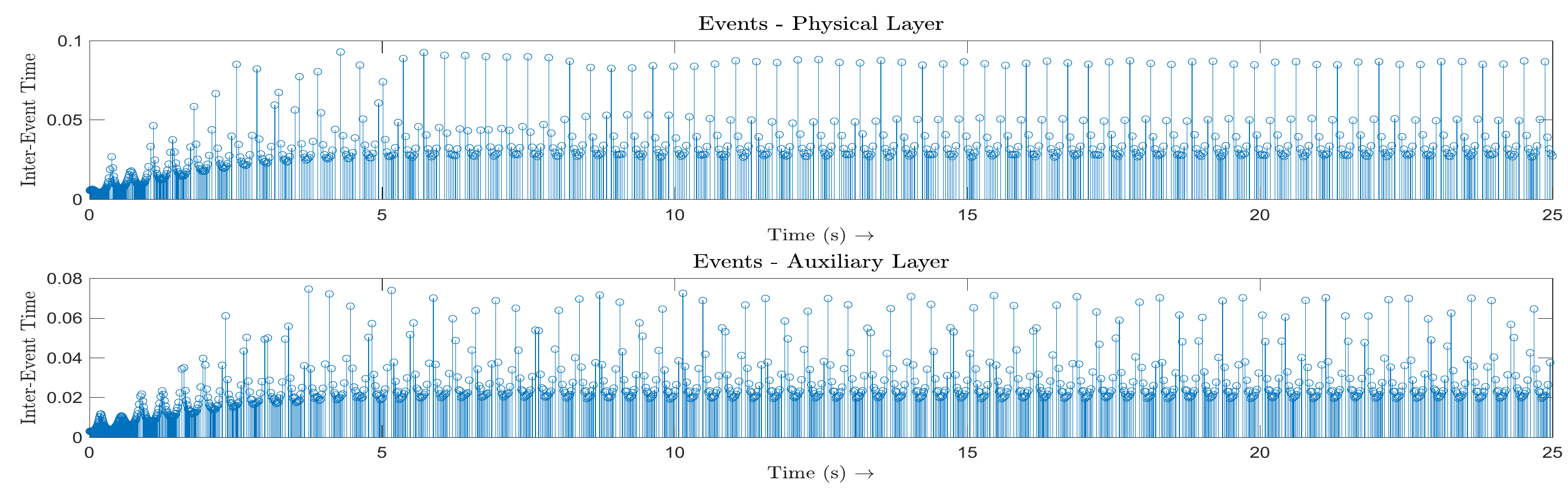}\label{fig:Sim_event25}}
	\caption{Physical and virtual agents' states, and the event sequence for system \eqref{err_system} with $\beta = 2.5$.}
	\label{fig:Sim_25}
	\vspace{-13pt}
\end{figure*}

\begin{figure}[t!]
	\centering
	\includegraphics[width=0.7\linewidth]{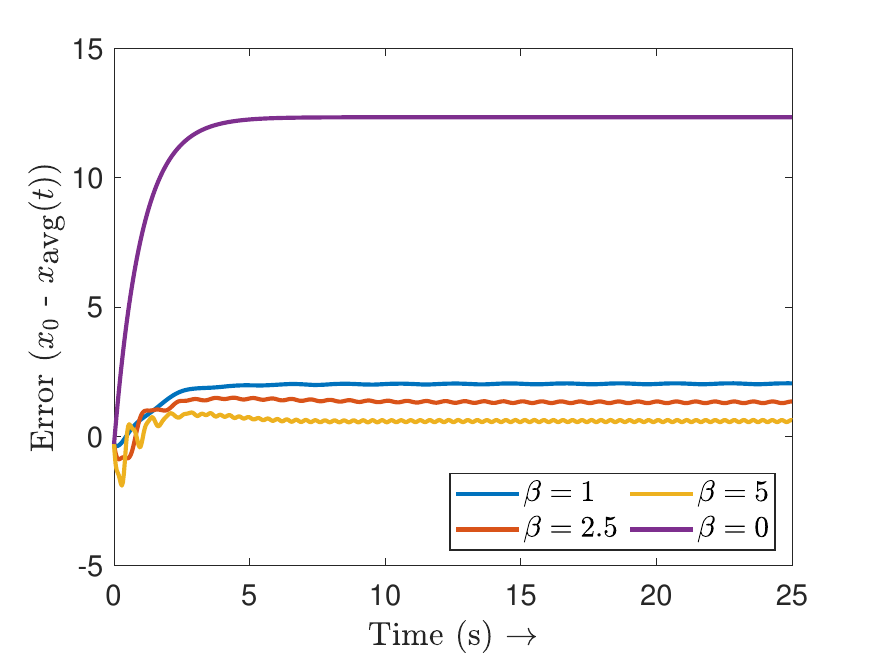}
	\caption{Error ($x_0 - x_{\text{avg}}$) between the leader and mean state for different values of $\beta$.}
	\label{fig:Error}
	\vspace{-15pt}
\end{figure}

\begin{remark}[Selection of inter-layer gain $\beta$]
	Simulation results indicate that large values of $\beta$ lead to state consensus closer to the leader's state (see Fig.~\ref{fig:Error}), however, with an increased number of events and a reduced MIET (see Fig.~\ref{fig:Sim_event1} and \ref{fig:Sim_event25}). This observation highlights a trade-off between achieving precise convergence near the leader's state and the associated communication burden, both influenced by $\beta$. Therefore, selecting appropriate $\beta$ is crucial to match the practical requirements as emphasized in the following subsection.	 
\end{remark}

\subsection{Experiments} \label{sec:Experiments}
For experiments, we use three \emph{Khepera IV}\footnote{Developed by K-Team Corporation, \url{https://www.k-team.com/khepera-iv}} differential drive robots operating within a motion capture (MoCap) system comprising overhead cameras (see Fig.~\ref{fig:HW}) and study the problem of consensus in their heading (velocity) directions. The differential drive model can be equivalently represented by the unicycle model as: $\dot{p}_i(t) = v_i(t)[\cos(\theta_i(t)), \sin(\theta_i(t))]^\top$ and $\dot{\theta}_i(t) = \omega_i(t)$, where $p_i(t) = [p_{x_i}(t), p_{y_i}(t)]^\top \in \mathbb{R}^2$ is the position, $\theta_i(t) \in [0, 2\pi)$ is the heading angle, and $v_i(t) \in \mathbb{R}$ and $\omega_i(t) \in \mathbb{R}$ are the linear and angular velocity controllers, respectively, for the $i^{\text{th}}$ robot, where $i = 1, 2, 3$. One of the simplest turn-rate controllers for the heading angle consensus is given by $\dot{\theta}_i(t) = \omega_i(t) = \sum_{j \in \mathcal{N}_i}(\theta_j(t) - \theta_i(t)) + b_i(\theta_0 - \theta_i), \forall i$ \cite{jin2017collision}, where $\theta_0$ is the leader heading angle. The preceding heading consensus model is equivalent to \eqref{system_sld} where we consider the attack signal $d(\theta, t)$ as described in the above simulation subsection~\ref{sec:Sim} with the same setting of the auxiliary layer parameters, and all-to-all communication topology among the three agents. In our experiments, we restrict the movement of the robots on the respective circles of (fixed) radii $r_i$ by controlling their linear speed $v_i(t)$ such that the relation $v_i(t) = r_i\omega_i(t)$ holds for all $t \geq 0$ for each $i$. In this situation, it can be observed that the problem of heading angle ($\theta_i$) consensus is equivalent to the problem of consensus in robots' angular positions ($x_i$) on their respective circles, where it holds that $\theta_i = x_i + ({\pi}/{2})$ for all $i$. With this setting, the attack signal $d$ affects both the angular velocity $\omega_i$, and hence, $v_i$, $\forall i$. Below, we describe the experiment setup and the procedure in a step-wise manner:

\begin{figure}[t!]
	\centering
	\subfigure[Khepera IV robot]{\includegraphics[width=0.15\textwidth]{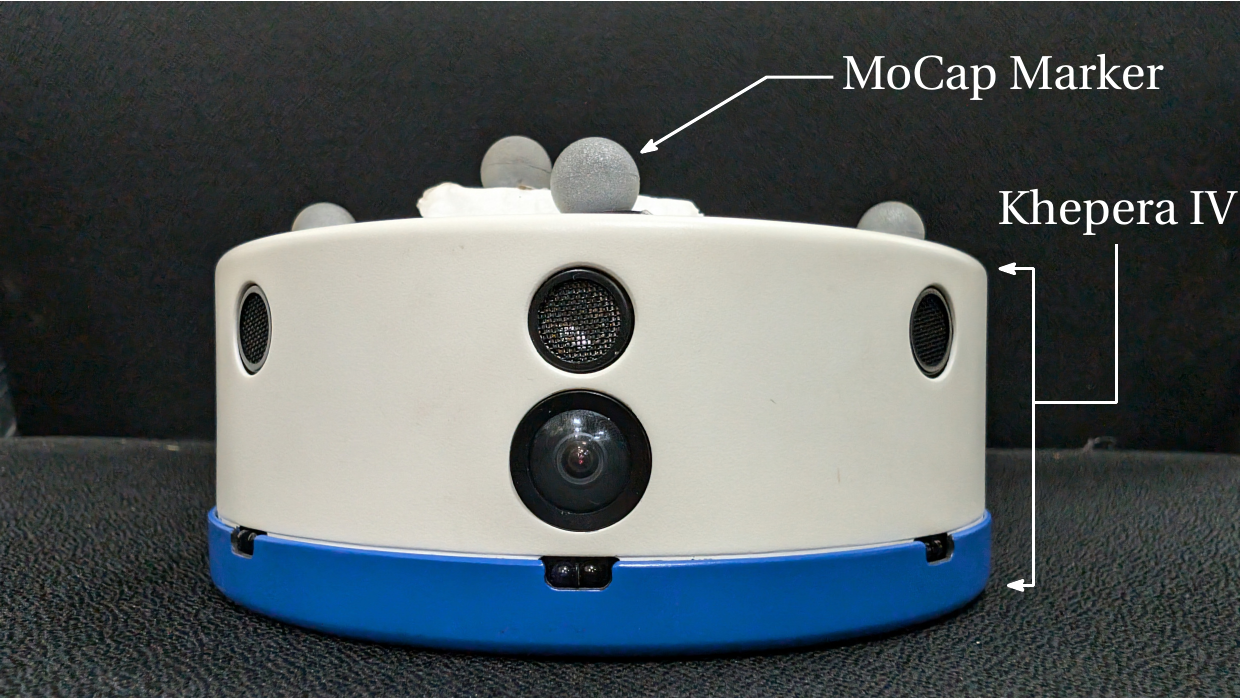}\label{fig:Khepera}} \hspace{5pt}
	\subfigure[MoCap cameras]{\includegraphics[width=0.15\textwidth]{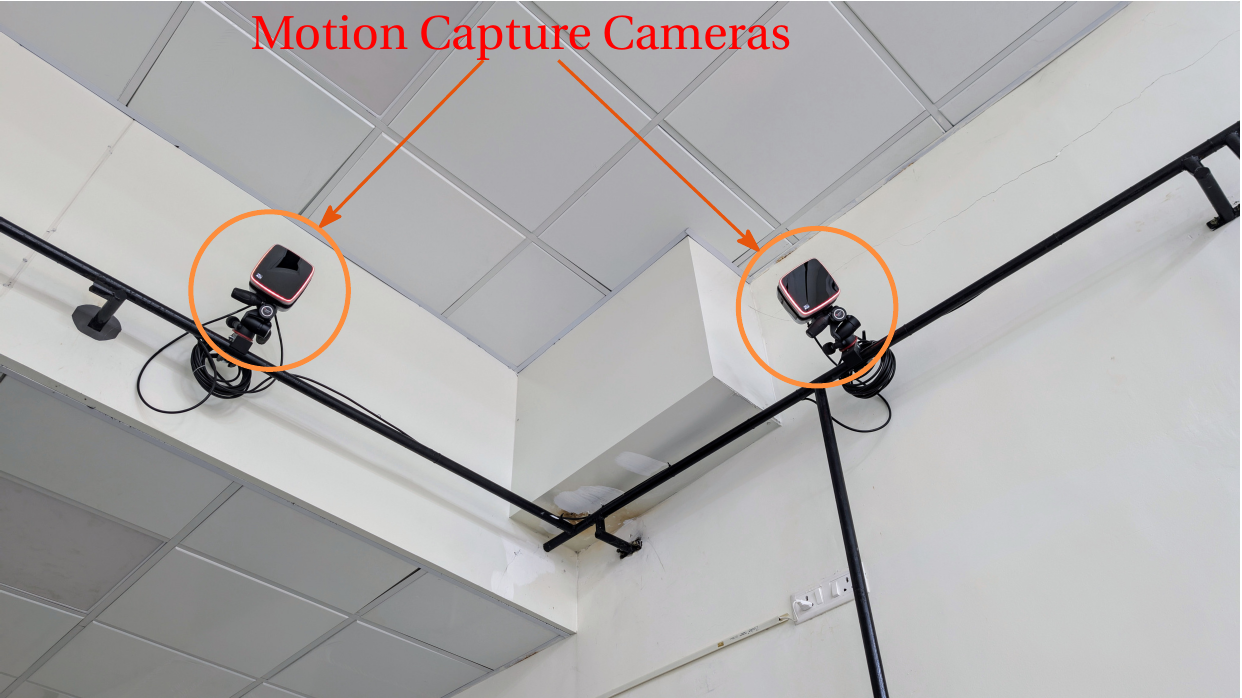}\label{fig:MoCap}}
	\caption{Khepera IV robot and MoCap setup.}
	\label{fig:HW}
	\vspace{-10pt}
\end{figure}

\begin{figure}[t!]
	\centering{
	\includegraphics[width=0.9\linewidth]{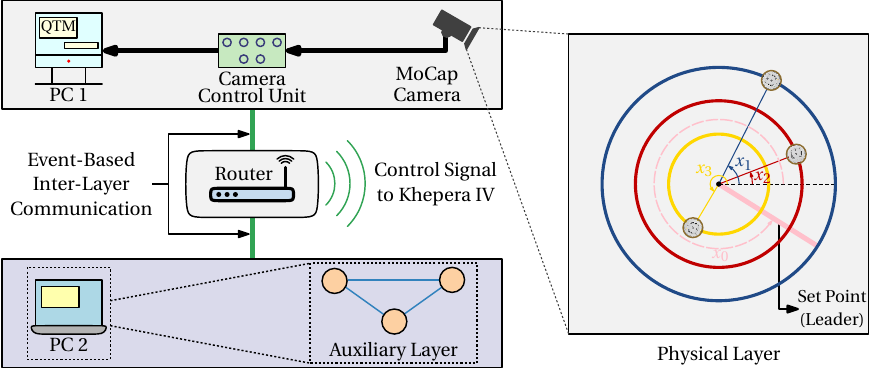}
	\caption{Schematic diagram for practical implementation.}
	\label{fig:HW_diag}}
	\vspace{-13pt}
\end{figure}

\begin{itemize}[leftmargin=*]
	\item The schematic diagram of the experimental setup is shown in Fig.~\ref{fig:HW_diag}, where the MoCap system collects the robot pose data (i.e., position and orientation) and shares it with the Qualisys Track Manager (QTM)\footnote{Developed by Qualisys, \url{https://www.qualisys.com/}}. The feedback from QTM is used to calculate the angular positions of the robots moving on their respective circles. The auxiliary layer dynamics are run on a separate computer using a Python program. The physical layer's data (PC1 and MoCap System) and the auxiliary layer's data (PC2) are shared through a router. The same router is used to transfer the control input to the Khepera IV robots over a wireless Wi-Fi network. The schematic shown here may look like a centralized control; however, it is implemented in a distributed fashion within the computer programs. 
	\item Before starting the experiment, the three Khepera IV robots are manually kept on the circles with tangential heading angles. The radii of the circles are taken as $r_1 = 0.5$ m (yellow), $r_2 = 0.9$ m (red), and $r_3 = 1.3$ m (blue). As shown in Fig.~\ref{fig:HW_diag}, the desired angular position set point (i.e., leader's state) is chosen as $x_0 = 2\pi - 0.5$ radians (or $331.3521^\circ$).
	\item To show that leader-follower consensus can be achieved with the auxiliary layer, an experimental result for system \eqref{err_system} in the absence of attack signal is shown in Fig.~\ref{fig:No_attack}. The camera captured result in Fig.~\ref{fig:No_attack_cam} shows the consensus of robots at the leader state $x_0$ (green dots on the pink line from the center). Due to the uncertainties in hardware implementation, it can be noticed in Fig.~\ref{fig:No_attack} that all robots do not exactly converge on the leader's state $x_0$. However, the error in their angular positions is small enough to be ignored for practical applications. Further, Fig.~\ref{fig:No_attack_QTM} shows the positions and headings of robots over time, obtained from the MoCap software - QTM.
	\item Fig.~\ref{fig:Expt_3} shows the experimental trajectory plot under the attacked scenario for $\beta=1.25$. The offset set by the attacker is $x_{ad} = 0.5$ radians (or $28.6479^\circ$). For demonstration purposes, the attack is simulated within the PC1, of course, with the controller being oblivious to the presence of the attack. The experiment starts from the initial configuration shown in Fig.~\ref{fig:expt_125_init}. Fig.~\ref{fig:expt_125_mid} shows the positions of the robots at an intermediate time $t= 29.7272 s$ when the robots are farthest from each other. Finally, the steady-state positions of the robots are shown in Fig.\ref{fig:expt_125_end}, where it can be observed that the robots stabilize \emph{near} the leader's angular position $x_0$. The events generated are recorded with the inter-event time and are plotted in Fig.~\ref{fig:Expt_event125}. 
\end{itemize}

\begin{figure}[t!]
	\centering
	\subfigure[From camera]{\includegraphics[width=0.3\linewidth]{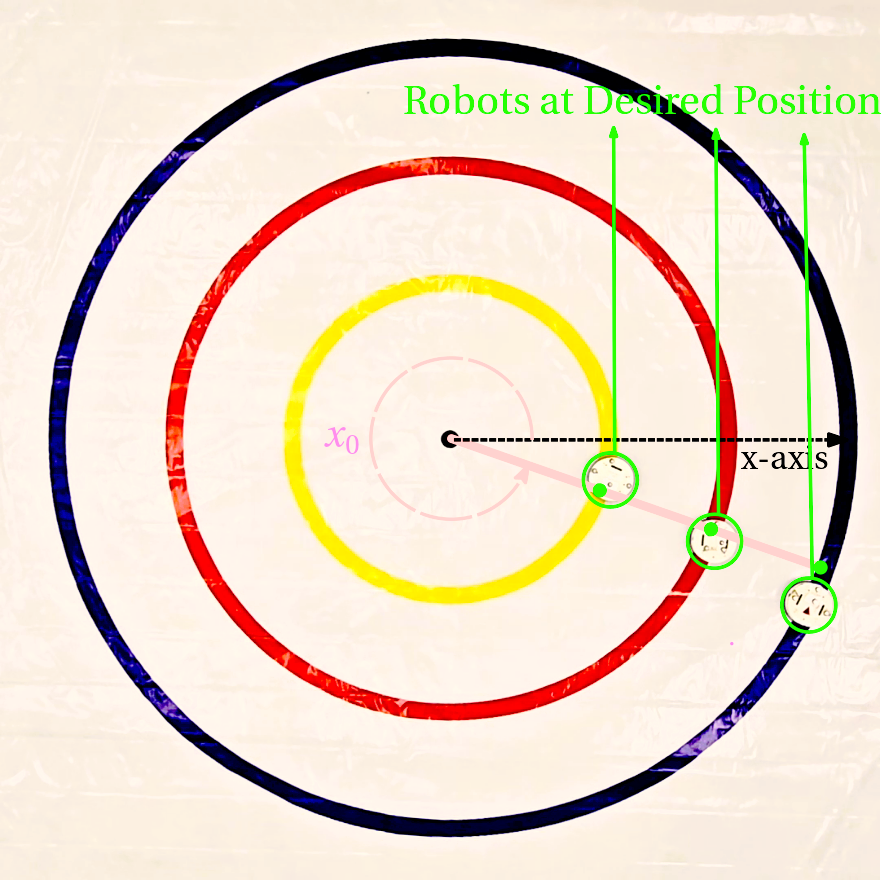}\label{fig:No_attack_cam}} \hspace{10pt}
	\subfigure[From QTM]{\includegraphics[width=0.3\linewidth]{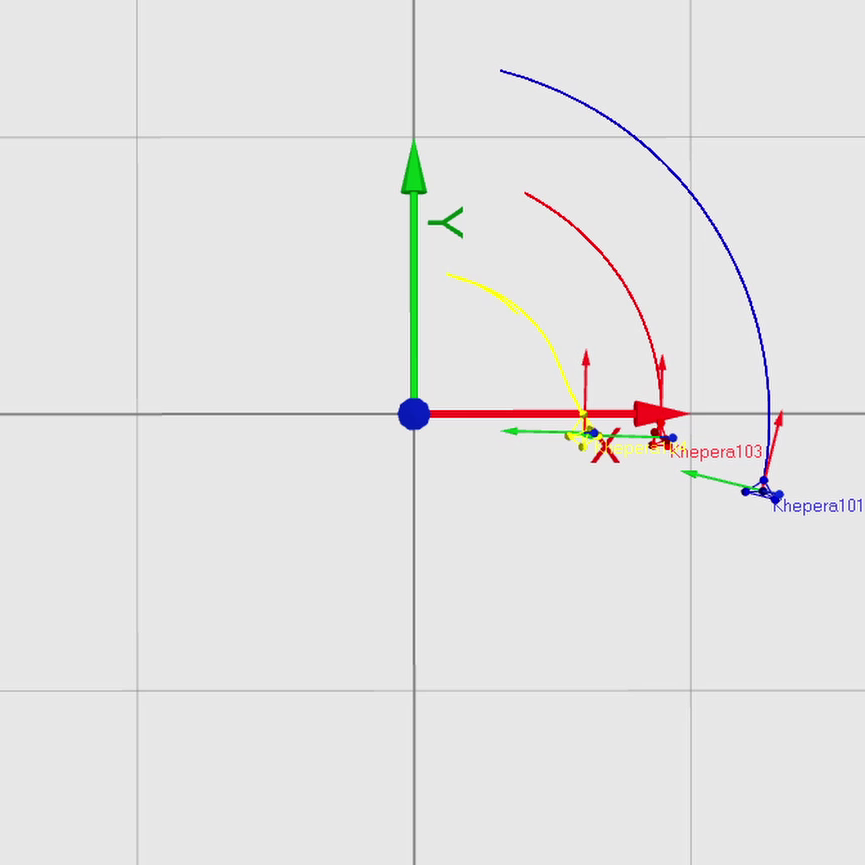}\label{fig:No_attack_QTM}}	
	\caption{Experimental results under no attack with $\beta = 1.0$.}
	\label{fig:No_attack}
	\vspace{-15pt}
\end{figure}

\begin{figure*}[t!]
	\centering{
		\subfigure[Initial heading]{\includegraphics[width=0.23\textwidth]{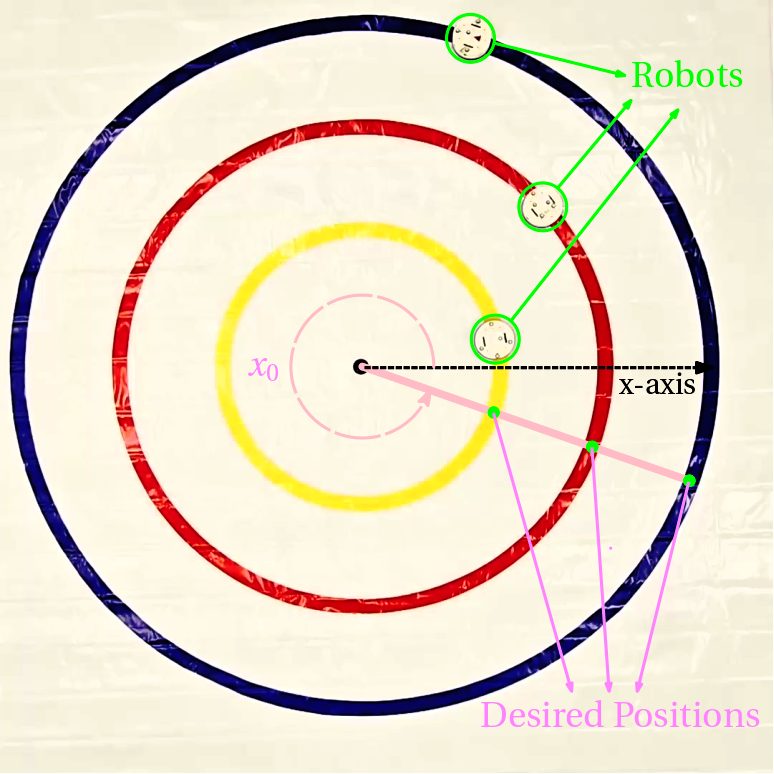}\label{fig:expt_125_init}} \hspace{10pt}
		\subfigure[Intermediate headings]{\includegraphics[width=0.23\textwidth]{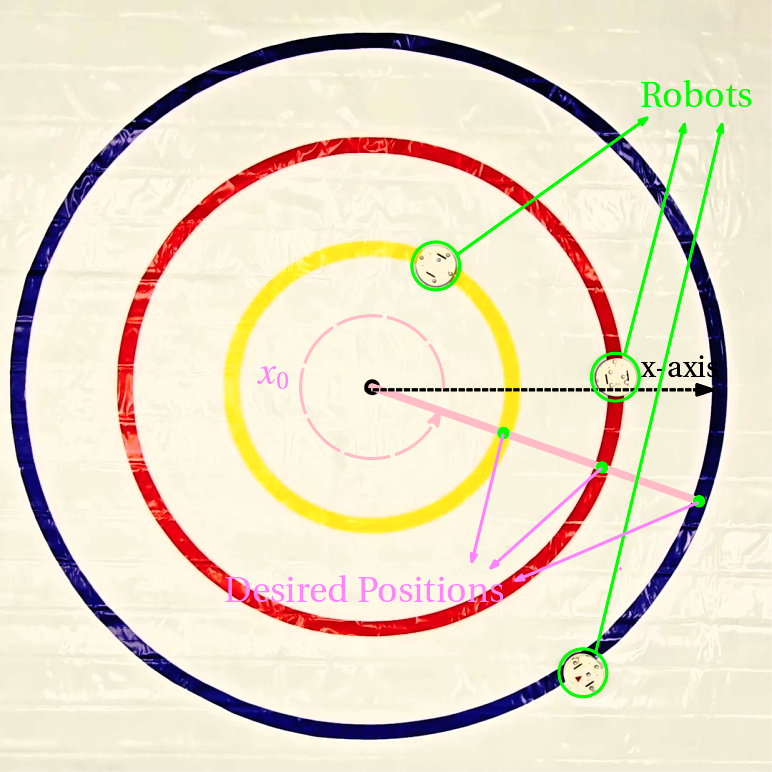}\label{fig:expt_125_mid}}\hspace{10pt}
		\subfigure[Final headings]{\includegraphics[width=0.23\textwidth]{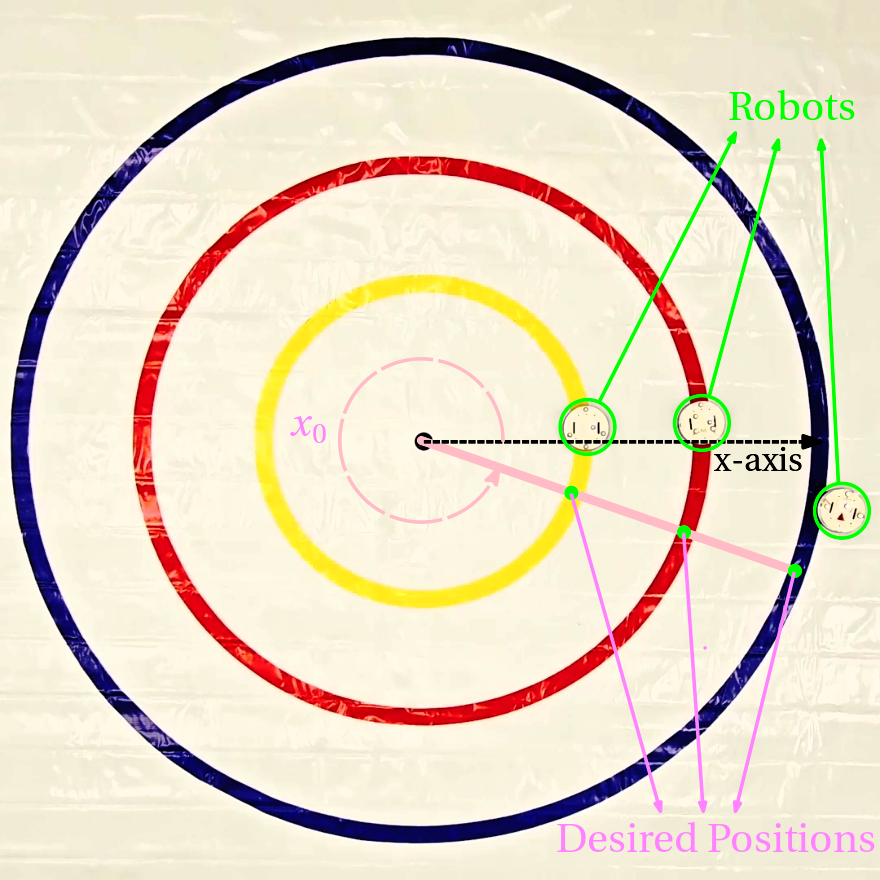}\label{fig:expt_125_end}} \hspace{10pt}
		\subfigure[Events generated in both layers]{\includegraphics[width=0.8\textwidth]{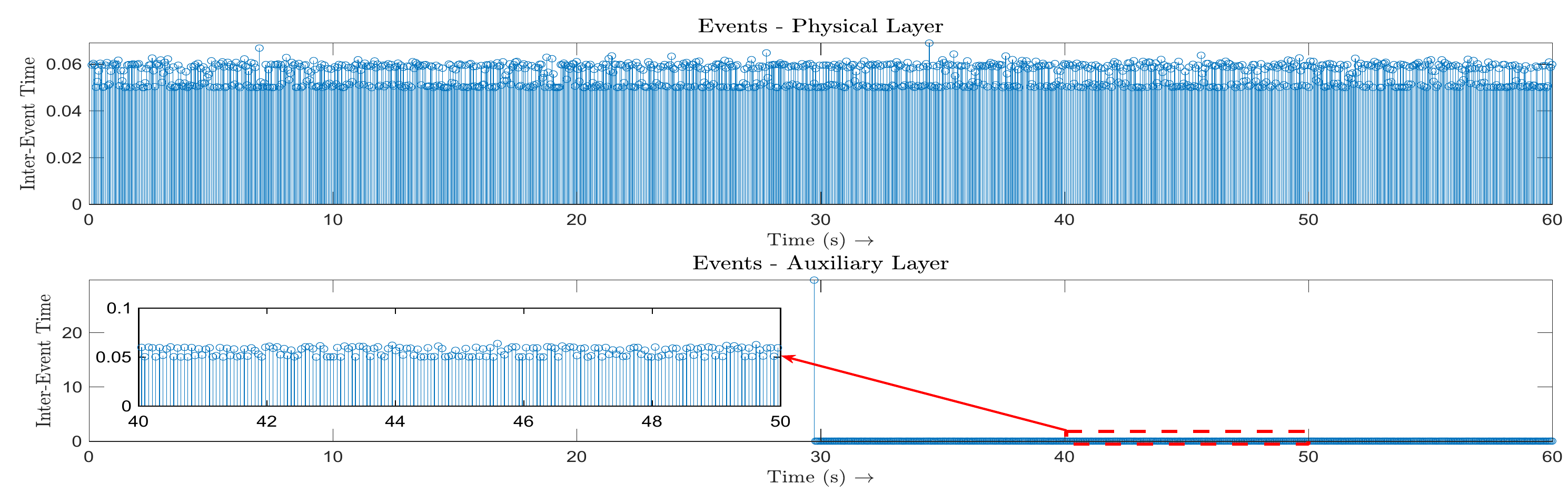}\label{fig:Expt_event125}}		
		\caption{Experimental results for heading angle consensus under attack with $\beta = 1.25$.}
		\label{fig:Expt_3}}
	\vspace{-15pt}
\end{figure*}

\section{Conclusions and Future Remarks} \label{sec:Conclusion}
Using an auxiliary layer, a leader-follower resilient state consensus problem was investigated in the presence of FDI attacks. To overcome the challenges posed by the increased communication burden between the layers, two event-triggering schemes for the communication between the agents in the auxiliary and physical layers were designed: (i) state-based event condition, where agents in both layers transmit the states at the same time $t_k$, and (ii) dynamic event condition, where agents in the physical layer transmit their states at time $t_{k_x}$ and agents in auxiliary layer transmit their state at time $t_{k_z}$. The absence of Zeno behavior was guaranteed by ensuring the presence of positive MIET in both cases, and the trade-off between the accurate state consensus and the number of events with inter-layer gain $\beta$ was discussed. The theoretical results obtained were verified with simulations, comparing different values of $\beta$. Finally, we discussed the implementation of an auxiliary layer scheme on hardware and carried out experiments on Khepera IV robots while restricting their motion on concentric circles. 

Besides, there are several challenges to the proposed approach that are worth considering in future research. Even though using an auxiliary layer to achieve resilience against communication attacks is easy to design and computationally moderate, the high oscillatory behavior of the system with event-triggered communication limits the practical applications to non-critical operations like equalized water irrigation and distributed optimization \cite{gusrialdi2023resilient}, where bound on consensus can be relatively large. Additionally, it is observed during experiments that larger values of inter-layer gain $\beta$ result in large values of $\dot{x}_i$, equivalently larger values of the Khepera IV robot's linear ($v_i$) and angular ($\omega_i$) velocities causing actuator saturation. Further, even though the presence of MIET is established, and thus the absence of Zeno behavior is guaranteed, it may happen that for large values of the inter-layer gain $\beta$, the system has practical Zeno behavior. Essentially, for large values of $\beta$, the MIET  might become smaller than the single cycle of the digital controller, thus resulting in inter-layer communication in every processor cycle.

\appendix

\section*{MIET for Single Event Condition} \label{MIET_single_event}
In this case, it is challenging to find an actual MIET due to the involvement of multiple sampling errors. However, it is possible to derive a non-conservative (large) bound on the MIET by considering the individual errors, as discussed below. At every event $t_k$, the sampling errors are reset, i.e., $e_x(t_k) = 0$ and $e_z(t_k)=0$. As per the event sequence \eqref{event2}, an event is triggered at MIET for the minimum value of the right side of inequality \eqref{event_cond_single_layer}, i.e.,
\begin{equation*}
	\frac{\beta}{c_1} \|P_x K\|^2 \|e_z\|^2 + \frac{\beta}{c_2} \|P_z G\|^2 \|e_x\|^2 = \kappa \left( \varepsilon - \frac{\|P_x\|^2}{c_3} \bar{D}^2 \right).
\end{equation*}
Analogous to above discussion before Theorem~\ref{thm_MIET}, we consider the agents' dynamics are bounded i.e., $\| \dot{x} \| \leq M_x$ and $\| \dot{z} \| \leq M_z$ for some positive constants $M_x$ and $M_z$. Hence, the sampling error $e_x = x - \bar{x}$, $e_z = z - \bar{z}$ and their derivative norms $\| \dot{e}_x\| = \| \dot{x} \| \leq M_x$ and $\| \dot{e}_z \| = \| \dot{z} \| \leq M_z$ are bounded. Furthermore, the rate of increase of the square of error norms are given by:
\begin{subequations}
	\begin{align} 
	\label{err_dynamics_x}\frac{d}{dt} \|e_x\|^2 &= \frac{d}{dt} e_x^\top e_x = 2 e_x^\top \dot{e_x} \leq 2\|e_x\| M_x\\
	\label{err_dynamics_z}
	\frac{d}{dt} \|e_z\|^2 &= \frac{d}{dt} e_z^\top e_z = 2 e_z^\top \dot{e_z} \leq 2\|e_x\| M_z.
	\end{align}
\end{subequations}
At the MIET ($\tau$), we have
\begin{multline} \label{event_MIET}
	\frac{\beta}{c_1} \|P_x K\|^2 \|e_z(t_k + \tau)\|^2 + \frac{\beta}{c_2} \|P_z G\|^2 \|e_x(t_k + \tau)\|^2 \\
	= \kappa \left( \varepsilon - \frac{\|P_x\|^2}{c_3} \bar{D}^2 \right).
\end{multline}
Considering a liberal bound on the MIET, we have the following conditions from \eqref{event_MIET}:
\begin{subequations} \label{err_tau}
	\begin{align}
		\frac{\beta}{c_2} \|P_z G\|^2 \|e_x(t_k + \tau)\|^2 & \leq \kappa \left( \varepsilon - \frac{\|P_x\|^2}{c_3} \bar{D}^2 \right) \label{err_tau_x}, \\
		\frac{\beta}{c_1} \|P_x K\|^2 \|e_z(t_k + \tau)\|^2 & \leq \kappa \left( \varepsilon - \frac{\|P_x\|^2}{c_3} \bar{D}^2 \right). \label{err_tau_z}
	\end{align}
\end{subequations}
Putting sampling error from \eqref{err_tau_x} and \eqref{err_tau_z} into \eqref{err_dynamics_x} and \eqref{err_dynamics_z}, respectively, we get:
\begin{subequations}
	\begin{align} 
\label{err_dynamics_x2} \frac{d}{dt} \|e_x\|^2 & \leq 2 M_x \sqrt{\frac{c_2 \kappa (\varepsilon - \frac{\|P_x\|^2 \bar{D}^2}{c_3})}{\beta\|P_z G\|^2}},\\
\label{err_dynamics_z2} \frac{d}{dt} \|e_z\|^2 & \leq 2 M_z \sqrt{\frac{c_1 \kappa (\varepsilon - \frac{\|P_x\|^2 \bar{D}^2}{c_3})}{\beta\|P_x K\|^2}}.
	\end{align}
\end{subequations}
Adding \eqref{err_dynamics_x2} and \eqref{err_dynamics_z2}, we get
\begin{equation}\label{intermediate_bound}
\frac{d}{dt} \|e_x\|^2 + \frac{d}{dt} \|e_z\|^2	\leq \rho,
\end{equation}
where $\rho {=} 2\left(M_x \sqrt{\frac{c_2 \kappa \left(\varepsilon {-} \frac{\|P_x\|^2 \bar{D}^2}{c_3}\right)}{\beta\|P_z G\|^2}} {+} M_z \sqrt{\frac{c_1 \kappa \left(\varepsilon {-} \frac{\|P_x\|^2 \bar{D}^2}{c_3}\right)}{\beta\|P_x K\|^2}}\right)$. Taking integral of \eqref{intermediate_bound} from $[t_k,t_k + \tau)$, yields
\begin{equation} \label{tau_s}
\|e_x(t_k + \tau)\|^2 + \|e_z(t_k +\tau)\|^2 \leq \rho \tau.
\end{equation}
Putting \eqref{err_tau_x} and \eqref{err_tau_z} into \eqref{tau_s}, implies
\begin{equation*}
	\frac{c_2 \kappa \left(\varepsilon - \frac{\|P_x\|^2 \bar{D}^2}{c_3}\right)}{\beta\|P_z G\|^2} + \frac{c_1 \kappa \left(\varepsilon - \frac{\|P_x\|^2 \bar{D}^2}{c_3}\right)}{\beta\|P_x K\|^2} \leq \rho \tau,
\end{equation*}
Consequently, the MIET is given by:
\begin{equation*} 
	\tau \geq \frac{1}{\rho} \left[\frac{c_2 \kappa \left(\varepsilon - \frac{\|P_x\|^2 \bar{D}^2}{c_3}\right)}{\beta\|P_z G\|^2} + \frac{c_1 \kappa \left(\varepsilon - \frac{\|P_x\|^2 \bar{D}^2}{c_3}\right)}{\beta\|P_x K\|^2}\right] > 0.
\end{equation*}

\bibliographystyle{IEEEtran}
\bibliography{References}

\end{document}